\documentclass[a4paper]{article}
\usepackage{mathtools}

\usepackage{enumitem}
\usepackage{calc}

\usepackage{lmodern} 

\usepackage{graphicx} 

\usepackage{amsmath}
\usepackage{amssymb,tikz}
\usepackage{pict2e}
\usepackage{amsthm}
\usepackage{amscd}
\usepackage{authblk}
\usepackage{mathrsfs}
\usepackage{todonotes}
\usetikzlibrary{calc} 

\usepackage{yfonts}

\usepackage{marvosym}
\usepackage[percent]{overpic}

\newtheorem{theorem}{Theorem} [section]
\newtheorem{proposition}[theorem]{Proposition}	
\newtheorem{corollary}[theorem]{Corollary}	
\newtheorem{lemma}[theorem]{Lemma}		

\newtheorem{assumptions}[theorem]{Assumptions}

\newtheorem{remark}[theorem]{Remark}

\newtheorem{examples}[theorem]{Examples}
\theoremstyle{definition}

\newcommand{\C}{\mathbb{C}}
\newcommand{\R}{\mathbb{R}}

\newcommand{\Ai}{{\rm Ai}}

\def\XXint#1#2#3{{\setbox0=\hbox{$#1{#2#3}{\int}$}
\vcenter{\hbox{$#2#3$}}\kern-.5\wd0}}

\usepackage{tikz}
\usetikzlibrary{arrows}
\usetikzlibrary{decorations.pathmorphing}
\usetikzlibrary{decorations.markings}
\usetikzlibrary{patterns}
\usetikzlibrary{automata}
\usetikzlibrary{positioning}
\usepackage{tikz-cd}
\tikzset{->-/.style={decoration={
				markings,
				mark=at position #1 with {\arrow{latex}}},postaction={decorate}}}
	
	\tikzset{-<-/.style={decoration={
				markings,
				mark=at position #1 with {\arrowreversed{latex}}},postaction={decorate}}}

\usetikzlibrary{shapes.misc}\tikzset{cross/.style={cross out, draw, 
         minimum size=2*(#1-\pgflinewidth), 
         inner sep=0pt, outer sep=0pt}}

\usepackage{pgfplots}
	
\numberwithin{equation}{section}

\usepackage{hyperref}
\hypersetup{colorlinks=true, urlcolor=blue, citecolor=red, linkcolor=blue}

\newcommand{\be}{\begin{equation}}
\newcommand{\ee}{\end{equation}}

\newcommand{\wt}{\widetilde}

\newcommand{\s}{\sigma}

\def\Tr{{\rm Tr}}
    \def\e{{\epsilon}}

    \def\Im{{\rm Im \,}}
    \def\Ai{{\rm Ai \,}}

    \def\Res{{\rm Res}}
    
    \def\P2n{{\rm P}_{{\rm II}}^{(n)}}
    \def\R{\mathbb{R}}
    \def\K{\mathcal{K}}
\def \O {\mathcal{O}}
\def \pa {\partial}
\def\Im{\mathrm{Im}\,}
\def\s{ {\sigma}}
\def\i{ {\mathrm{i}}}
\def\e{{\mathrm{e}}}
\def\1{\mathbf{1}}
\def\C{\mathbb{C}}
\def\R{\mathbb{R}}
\def\g{\gamma} 
\def\d{\mathrm{d}}
\def\z{\zeta}
\def\wt{\widetilde}

\def\G{\Gamma}
\def\Ai{\mathrm{Ai}}
\def\Airy{{\rm Ai}}
\def\Bessel{{\rm Be}}
\def\K{{\mathbb K}}
\def\L{{\mathbb L}}
\def\step{{\sigma_\gamma}}
\begin{document}
\title{Airy kernel determinant solutions to the KdV equation
and integro-differential Painlev\'e equations}
\author[1]{Mattia Cafasso}
\author[2]{Tom Claeys}
\author[3]{Giulio Ruzza}
\renewcommand\Affilfont{\small}
\affil[1]{\textit{LAREMA, UMR 6093, UNIV Angers, CNRS, SFR Math-STIC, France;} \texttt{cafasso@math.univ-angers.fr}}
\affil[2,3]{\textit{Institut de Recherche en Math\'ematique et Physique,  UCLouvain, Chemin du Cyclotron 2, B-1348 Louvain-la-Neuve, Belgium;} \texttt{tom.claeys@uclouvain.be}, \texttt{giulio.ruzza@uclouvain.be}}
\maketitle

\begin{abstract}
We study a family of unbounded solutions to the Korteweg-de Vries equation which can be constructed as log-derivatives of deformed Airy kernel Fredholm determinants, and which are connected to an integro-differential version of the second Painlev\'e equation.
The initial data of the Korteweg-de Vries solutions are well-defined for $x>0$, but not for $x<0$, where the solutions behave like $\frac{x}{2t}$ as $t\to 0$, and hence would be well-defined as solutions of the cylindrical Korteweg-de Vries equation.
We provide uniform asymptotics in $x$ as $t\to 0$; for $x>0$ they involve an integro-differential analogue of the Painlev\'e V equation. 
A special case of our results yields improved estimates for the {tails} of the narrow wedge solution to the Kardar-Parisi-Zhang equation.
\end{abstract}

{
\hypersetup{linkcolor=black}
\tableofcontents
}

\section{Introduction and statement of results}

The Airy kernel is defined as
\begin{equation}\label{Airykernel}
K^{\Ai}\left(u,v\right)=\frac{\Ai\left(u\right)\Ai'\left(v\right)-\Ai'\left(u\right)\Ai\left(v\right)}{u-v},
\end{equation}
where $\Ai$ is the Airy function.
It is the correlation kernel characterizing the Airy point process, which is a determinantal point process on the real line, whose random point configurations almost surely are infinite and have a largest particle.
Denoting 
 $\zeta_1 > \zeta_2 > \ldots > \zeta_k > \ldots$ 
for the points in the configuration, averages of multiplicative statistics are Fredholm determinants: for a large class of functions $\phi$, we have \cite{Johansson,L1,L2,L3,Soshnikov}
\begin{equation}\label{AiryFredholm}
\mathbb E_\Ai\left[\prod_{j=1}^\infty\left(1-\phi\left(\zeta_j\right)\right)\right]=\det\left(1- \K_\phi^{\Ai}\right),
\end{equation}
where the right hand side is the Fredholm determinant of the identity operator minus the integral operator acting on $L^2\left(\mathbb R\right)$ with kernel $\sqrt{\phi\left(u\right)}K^\Ai\left(u,v\right)\sqrt{\phi\left(v\right)}$,
\begin{equation}\label{def:Fredholm}
	\det\left(1-\K_\phi^{\Ai}\right) = \sum_{k = 0}^\infty \frac{\left(-1\right)^k}{k!} \int_{\mathbb R^k} \det\left(\phi\left(u_i\right) K^{\Ai}\left(u_i,u_j\right)\right)_{i,j = 1}^k \d u_1\ldots \d u_k.
\end{equation}
In order to gain some intuition about the behavior of \eqref{AiryFredholm}, it is worth mentioning that the average number of particles larger than $x$ in the Airy point process is equal to $\int_x^{+\infty}K^\Ai(u,u)\d u$, and by the properties of the Airy function this integral behaves like $\frac{2}{3\pi}|x|^{3/2}$ as $x\to -\infty$, see e.g. \cite{Soshnikov}.

In this paper, we will study a specific type of averages of the form \eqref{AiryFredholm}. Given a piecewise smooth weakly increasing function $\sigma : \mathbb R \longrightarrow [0,1]$ satisfying Assumptions \ref{assumptions} below, we define for any $x \in \mathbb R$ and $t > 0$, 
\begin{equation}\label{def:Q}
	Q_\sigma\left(x,t\right) := \mathbb E_\Ai\left[\prod_{j=1}^\infty\left(1-\sigma\left(t^{-2/3}\zeta_j+x/t\right)\right)\right]=\det\left(1- \K^\Ai_{\sigma, x, t}\right),
\end{equation}
where $\K^\Ai_{\sigma,x,t}$ is the integral operator with kernel
\begin{equation}\label{def:kernel}
K_{\sigma, x, t}\left(u,v\right)= \sqrt{\sigma\left(t^{-2/3}u+x/t\right)}K^\Ai\left(u,v\right) \sqrt{\sigma\left(t^{-2/3}v+x/t\right)}.
\end{equation}
We will show that $u(x,t)=\partial_x^2\log Q_\sigma(x,t)+\frac{x}{2t}$  solves the Korteweg-de Vries (KdV) equation
\begin{equation*}
 u_t + 2 u u_x + \frac{1}6 u_{xxx} = 0,
\end{equation*}
and we will study the small $t$ asymptotics of these solutions. Moreover, a Riemann-Hilbert (RH) characterization of this family of KdV solutions will reveal a connection with the integro-differential Painlev\'e II equations introduced by Amir, Corwin, and Quastel in \cite{AmirCorwinQuastel}.

The function $\sigma$ needs to satisfy the following assumptions:

\begin{assumptions}\label{assumptions}\hfill
\begin{enumerate}
\item $\sigma:\mathbb R\to[0,1]$ is non-decreasing and it is $C^\infty$ except possibly at a finite number of points $r_1, \ldots, r_k$; we denote $\gamma=\lim_{r\to +\infty}\s(r)\in[0,1]$,
\item the left and right limits of $\sigma$ and of all of its derivatives exist at the points $r_1,\ldots, r_k$ but are in general not the same,
\item the function $r\mapsto r^2 \sigma(r)$ is in $L^2(\mathbb R^-, \d r)$.
\end{enumerate}
\end{assumptions}

\begin{remark}\label{remark:traceclassandpositivity}
Since $\sigma(r)\leq 1$, the third condition
is nothing more than a decay condition for $\sigma(r)$ as $r\to -\infty$, satisfied if, for some $\epsilon>0$, $\s(r)=\mathcal O\left(|r|^{-5/2-\epsilon}\right)$ as $r\to -\infty$. 
This condition implies in particular that for any
 $x\in\mathbb R$, $t>0$, $\sigma\left(t^{-2/3}u+x/t\right)K^\Ai\left(u,u\right)\in L^1\left(\mathbb R,\d u\right)$ and hence that the (non-negative) operator $\K^\Ai_{\sigma,x,t}$ is trace-class, such that $Q_\sigma\left(x,t\right)$ is defined for any $x\in\mathbb R$ and $t>0$. Indeed, this follows from the inequalities
\begin{align*}\left\|\sigma\left(t^{-2/3}u+x/t\right)K^\Ai\left(u,u\right)\right\|_{L^1(\mathbb R, \d u)}
&\leq c\left(\left\|\sigma\left(t^{-2/3}u+x/t\right)(|u|+1)^{1/2}\right\|_{L^1(\mathbb R^-, \d u)}+1\right) 
\\ &\leq \widehat c(x,t)\left(\left\|r^2\sigma(r)\right\|_{L^2(\mathbb R^-, \d r)}.\left\|(|u|+1)^{-3/2}\right\|_{L^2(\mathbb R^-, \d u)}+1\right),\end{align*}
for large enough $c, \widehat c(x,t)>0$,
where we used first the $u\to \pm\infty$ asymptotics of $K^{\rm Ai}(u,u)$ and then the Cauchy-Schwarz inequality.

Moreover, assumptions \ref{assumptions} imply the strict inequality $Q_\sigma\left(x,t\right)>0$ for any $x\in\mathbb R$, $t>0$.
To see this, we only need to show that $1$ is not an eigenvalue of $\K^\Ai_{\sigma,x,t}$.
Indeed, we claim that if $f\in L^2(\mathbb R)$ satisfies $\K^\Ai_{\sigma,x,t}f=f$ then $f=0$.
To this end, let us denote $\mathbb M$ for the multiplication operator on $L^2(\R)$ given by $\bigl(\mathbb M \phi\bigr)(u)=\sqrt{\sigma(t^{-2/3}u+x/t)}\phi(u)$, such that $\K^\Ai_{\sigma,x,t}=\mathbb M\mathbb K^\Ai\mathbb M$ where $\mathbb K^\Ai$ is the integral operator with kernel $K^\Ai(u,v)$.
Coming back to our claim, let us define $g=\mathbb K^\Ai\mathbb Mf$.
By $\mathbb K^\Ai_{\sigma,x,t}f=f$ and by standard properties of the Airy function, $g$ is (the restriction to the real line of) an entire function and $f=\mathbb M g$.
On the other hand, it is well-known that $\mathbb K^\Ai$ is an orthogonal projector and therefore we get the following chain of inequalities
\be
\|g\|=\|\mathbb K^\Ai\mathbb M f\|\leq \|\mathbb M f\|\leq\|f\|=\|\mathbb M g\|\leq \|g\|
\ee
which collapses to a chain of equalities, were $\|\cdot\|$ is the $L^2(\R)$-norm and we use that $\sigma$ takes values in $[0,1]$.
Since $\sigma$ decays at infinity and $g$ is entire, the equality $\|\mathbb M g\|=\|g\|$ forces $g$ to vanish identically, hence $f=\mathbb Mg=0$ as claimed.
\end{remark}

Our first result is the following relation between the multiplicative statistics $Q_\sigma\left(x,t\right)$ and the KdV equation.
\begin{theorem}\label{theorem:main}
Let $\sigma$ satisfy Assumptions \ref{assumptions}, and let $Q_\sigma$ be as in \eqref{def:Q}.
Then, for any $x\in\mathbb R$ and $t>0$, the function
\begin{equation}\label{def:uQ}
	u_\sigma\left(x,t\right):=\partial_x^2\log Q_\sigma\left(x,t\right)+\frac{x}{2t}
\end{equation}
solves the KdV equation
\begin{equation}
\label{eq:KdV}
\partial_t u_\sigma + 2 u_\sigma \partial_x u_\sigma + \frac{1}6 \partial_x^3 u_\sigma = 0.
\end{equation}
Moreover, it can be expressed as 
	\begin{equation}\label{eq:u2}
		u_\sigma\left(x,t\right)=-\frac{1}{t}\int_{\mathbb R} \phi_\sigma^2\left(r;x,t\right)\d \s(r)+\frac{x}{2t},
	\end{equation}
where $\phi_\sigma$ solves the Schr\"odinger equation with potential $2u_\sigma$,
\begin{equation}\label{eq:Schrodinger}
\partial_x^2\phi_\sigma\left(z;x,t\right)=\left(z - 2u_\sigma\left(x,t\right)\right)\phi_\sigma\left(z;x,t\right),\qquad z\in\R,\ x\in\R,\ t>0,
\end{equation}
the evolution equation
\begin{equation}\label{eq:eveq}
	\partial_t \phi_\sigma\left(z;x,t\right) = -\left(\frac{2}3 z + \frac{2}3 u_\sigma \left(x,t\right) \right) \partial_x \phi_\sigma\left(z;x,t\right) + \frac{1}3 \partial_x u_\sigma(x,t)\phi_\sigma\left(z;x,t\right)
\end{equation}
and has asymptotic behavior 
\begin{equation}\label{eqthm:phiAi}
 \phi_\sigma\left(z;x,t\right)\sim t^{1/6}\Ai\left(t^{2/3}z-xt^{-1/3}\right),\qquad\mbox{as $z\to \infty$ with $|\arg z|<\pi-\delta$, for any $\delta>0$,}
\end{equation}
pointwise in $x\in\R$ and $t>0$.
\end{theorem}

The integral with respect to $\d\s(r)$ in \eqref{eq:u2} is to be understood as a Riemann-Stieltjes integral. Equivalently, $\d\s(r)$ denotes the measure $\s'(r)\d r+\sum_{j=1}^km_j\delta_{r_j}(\d r)$, where $r_1,\dots,r_k$ are the only singularities of $\s$, $m_j:=\lim_{\epsilon\to 0_+}\left(\s(r_j+\epsilon)-\s(r_j-\epsilon)\right)$ and $\delta_r$ denotes the Dirac delta measure supported at $r\in\R$.

As an immediate consequence of Theorem \ref{theorem:main}, we obtain the following corollary by combining \eqref{def:uQ}, \eqref{eq:u2}, and \eqref{eq:Schrodinger}; modulo a rescaling of the variable, this corollary is equivalent to equations (323) and (324) in Proposition 40 of \cite{AmirCorwinQuastel}.

\begin{corollary}
Let $\sigma$ satisfy Assumptions \ref{assumptions}, and let $Q_\sigma$ be as in \eqref{def:Q}.
Then
	\begin{equation*}
		\partial_x^2 \log Q_\sigma\left(x,t\right) = -\frac{1}{t}\int_\mathbb R \phi_\sigma^2\left(r;x,t\right)\d\s(r),
	\end{equation*}
where $\phi_\sigma$ satisfies the integro-differential Painlev\'e II equation
	\begin{equation}\label{integrodiffPII}
	\partial_x^2\phi_\sigma\left(z;x,t\right) = \left(z - \frac{x}t + \frac{2}t \int_{\mathbb R} \phi_\sigma^2\left(r;x,t\right) \d\s(r) \right) \phi_\sigma\left(z;x,t\right).
	\end{equation}
\end{corollary}

\begin{remark}
In addition to the integro-differential Painlev\'e II equation above, one can also prove that the rescaled wave function $\hat\phi_\sigma(z;t_1,t_3)$ defined by
\begin{equation}\label{eq:mKdVchangevariables}
	\phi_\sigma(z;x,t) = t^{-1/2}\hat\phi_\sigma(z;t_1,t_3), \quad \text{with} \quad t_1 = -xt^{-1} \;\; \text{and} \;\; t_3 = \frac{1}3t^{-2}
\end{equation}
satisfies the integro-differential modified KdV equation\footnote{For simplicity, in the following equations of this remark we omit the dependence of $\hat\phi_\sigma(z;t_1,t_3)$ on the parameters $t_1$ and $t_3$.}
\begin{equation}
\label{eq:IDmKdV}
\partial_{t_3}\hat\phi_\sigma(z) = -\partial_{t_1}^3\hat\phi_\sigma(z) + 3 \partial_{t_1}\hat\phi_\sigma(z)\int_{\mathbb R}\hat\phi_\sigma^2(r)\mathrm{d}\sigma(r) + 3\hat\phi_\sigma(z)\int_{\mathbb R}\hat\phi_\sigma(r)\partial_{t_1}\hat\phi_\sigma(r)\mathrm{d}\sigma(r). 
\end{equation}
This equation has been recently found in \cite{BothnerCafassoTarricone} (see equation (25) in loc. cit.) as the first member of a whole hierarchy associated to a class of operator valued Riemann-Hilbert problems. 
In order to prove \eqref{eq:IDmKdV} within the setting of this paper, it is enough to start with equation \eqref{eq:eveq}, use the $x$-derivative of the equation \eqref{eq:Schrodinger} to eliminate the term in $z$ and finally use \eqref{eq:u2} (and its $x$-derivative) to write everything in terms of $\phi_\sigma$:
$$
	-\partial_t\phi_\sigma(z) = \frac{2}3\partial_x^3 \phi_{\sigma}(z) + \frac{\phi_\s(z)}{2t} + \frac{x \partial_x\phi_\sigma(z)}t - \frac{2}t \phi_\sigma(z) \int_{\mathbb R} \phi_\sigma(r) \partial_x\phi_\sigma(r)\d\sigma(r) -\frac{2}t \partial_x\phi_\sigma(z)\int_{\mathbb R} \phi_\sigma^2(r) \mathrm d \sigma (r).
$$
Equation \eqref{eq:IDmKdV} is finally obtained from the equation above after the change of variables \eqref{eq:mKdVchangevariables}.
\end{remark}

\begin{examples}\label{examples}
\upshape
{\it 1.} When $\sigma=\sigma_\gamma:= \gamma \chi_{\left[0,\infty\right)}$ with $\gamma\in\left(0,1\right]$, $Q_{\sigma_\gamma}\left(x,t\right) = F_{\rm TW}\left(-xt^{-1/3};\gamma\right)$ is the  Tracy-Widom distribution, which describes the limit law of the largest eigenvalue of a wide variety of random matrices for $\gamma=1$ and arises also in combinatorial problems such as the length of the longest increasing subsequence of a random permutation \cite{BaikDeiftJohansson}, as well as in several statistical physics models. 
The celebrated result of Tracy and Widom \cite{TW} implies that
\begin{equation}\label{eq:diffidPII}\partial_x^2\log Q_{\sigma_\gamma}\left(x,t\right)=-t^{-2/3}y_\gamma^2\left(-xt^{-1/3}\right),\end{equation} where $y_\gamma$ is the unique solution of the Painlev\'e II equation 
\begin{equation}\label{eq:PII}y_\gamma''\left(x\right)=xy_\gamma\left(x\right)+2y_\gamma^3\left(x\right)\ \mbox{ with asymptotic behavior $y_\gamma\left(x\right)\sim \sqrt{\gamma}\,{\rm Ai}\left(x\right)$ as $x\to +\infty$.}\end{equation} This solution is known as the Hastings-McLeod solution \cite{HastingsMcLeod} if $\gamma=1$ and as an Ablowitz-Segur solution \cite{AblowitzSegur} otherwise.
This is indeed consistent with Theorem \ref{theorem:main}:
\eqref{eq:u2} then becomes
\[u_{\sigma_\gamma}\left(x,t\right)=-\frac{\gamma}{t}\phi_{\sigma_\gamma}^2\left(0;x,t\right)+\frac{x}{2t},\]
and the Schr\"odinger equation \eqref{eq:Schrodinger} is
\[\partial_x^2\phi_{\sigma_\gamma}\left(0;x,t\right)=\left(-\frac{x}{t}+\frac{2\gamma}{t}\phi_{\sigma_\gamma}^2\left(0;x,t\right)\right)\phi_{\sigma_\gamma}\left(0;x,t\right),\] which turns into the Painlev\'e II equation after setting $y\left(x\right)=\g^{1/2} t^{-1/6}\phi_{\sigma_\gamma}\left(0;-xt^{1/3},t\right)$.
Conversely, it is straightforward to verify from the Painlev\'e II equation that the function 
\[u_\sigma\left(x,t\right)=\frac{x}{2t}-t^{-2/3}y_\gamma^2\left(-xt^{-1/3}\right)\] solves the KdV equation \eqref{eq:KdV}.

{\noindent \it 2.} In the case where $\sigma$ is a piecewise constant function with $m$ discontinuities, $Q_\sigma$ has been studied in \cite{CharlierClaeys,ClaeysDoeraene}. In particular, the integro-differential Painlev\'e II equation \eqref{integrodiffPII} then turns into a system of coupled Painlev\'e II equations, see also \cite{LyuChen, WuXu, XuDai}.

{\noindent \it 3.} Denoting
\begin{equation}
\label{sigmaKPZ}
\sigma_{\rm KPZ}\left(r\right):= \frac 1{1 + \exp\left(-r\right)},
\end{equation} 
we have that $Q_{\sigma_{\rm KPZ}}\left(sT^{-1/6},T^{-1/2}\right)$ characterizes the distribution of the solution of the Kardar-Parisi-Zhang (KPZ) equation with narrow wedge initial condition \cite{AmirCorwinQuastel, BorodinGorin} (see also \cite{CLDR, Dotsenko, SasamotoSpohn}); it appears also in a model of finite temperature free fermions (which is equivalent to the MNS matrix model) \cite{Johansson, LiechtyWang,MNS} and in the edge scaling limit of the periodic Schur process \cite{BeteaBouttier}.  
{More precisely, write $\mathcal H\left(T,X\right)$ for the solution of the KPZ equation	
\begin{equation}\label{eq:KPZ}
\partial_T\mathcal H\left(T,X\right)=\frac{1}{2}\partial_X^2\mathcal H\left(T,X\right)+\frac{1}{2}\left(\partial_X\mathcal H\left(T,X\right)\right)^2+\xi\left(T,X\right),
\end{equation}
where $\xi\left(T,X\right)$ denotes space-time white noise,
with narrow wedge initial condition $\mathcal H\left(0,X\right)=\log Z\left(0,X\right)$ with $Z\left(0,X\right)=\delta_{X=0}$ the Dirac distribution (see e.g.\ \cite{Corwin} for a rigorous definition of this solution), then for any $s\in\mathbb R$, $T>0$, we have the identity 
\begin{equation}
\label{identityKPZ}
\mathbb E_{\rm KPZ}\left[{\rm e}^{-{\rm e}^{\mathcal H\left(2T,X\right)-\frac{X^2}{4T}+\frac{T}{12}+sT^{1/3}}}\right]=Q_{\sigma_{\rm KPZ}}\left(sT^{-1/6},T^{-1/2}\right),
\end{equation}
where $\mathbb E_{\rm KPZ}$ denotes the expectation with respect to the KPZ equation.}
	Multiplicative statistics of {a similar nature for} more general point processes appeared recently in relation with the O'Connell-Yor semi-discrete Brownian directed polymer \cite{ImamuraSasamoto}, as well as in the stochastic six vertex model \cite{BorodinSixVertex}, see \cite{BorodinGorin} and references therein for more details. 

We note	that $\s_{\rm KPZ}$ can be extended to a meromorphic function in the complex plane, which has the periodicity property $\s_{\rm KPZ}(z)=\s_{\rm KPZ}(z+2\pi\i)$. This implies that we can also define $Q_{\s_{\rm KPZ}}(x+2\pi\i t,t)$ and $u_{\sigma_{\rm KPZ}}(x+2\pi\i t,t)$ for any $x\in\mathbb R$, $t>0$, and moreover we have
\[Q_{\s_{\rm KPZ}}(x+2\pi\i t,t)=Q_{\s_{\rm KPZ}}(x,t),\qquad u_{\s_{\rm KPZ}}(x+2\pi\i t,t)=u_{\s_{\rm KPZ}}(x,t)+\i\pi.\]
Setting $t=0$ in this identity yields a contradiction, and this is a first indication that the $t\to 0$ limit for the KdV solution $u_\sigma(x,t)$ is delicate. We confirm this below for general $\sigma$.
\end{examples}

\begin{remark}\label{remark:cylKdV}
Consider the change of variables $T = t^{-2},\ \rho = -{x}/t$ between $x\in\R,\ t>0$ and $\rho\in\R,\ T>0$ in \eqref{eq:u2}.
It is straightforward to verify that, since $u_\sigma\left(x,t\right) := \partial_x^2 \log Q_\sigma\left(x,t\right) + \frac{x}{2t}$ solves the KdV equation \eqref{eq:KdV}, the function {$$U_\sigma\left(\rho,T\right) := \partial_\rho^2 \log Q_\sigma\left(-\rho T^{-1/2},T^{-1/2}\right) = T^{-1} u_\sigma\left(-\rho T^{-1/2},T^{-1/2}\right) + \frac {\rho T^{-1}} 2  $$}
solves the cylindrical KdV equation
\begin{equation}\label{cylindricalKdV}
	\partial_T U_\sigma + \frac{1}{12}\partial_\rho^3U_\sigma + U_\sigma\partial_\rho U_\sigma + \frac{U_\sigma}{2T} =  0.
\end{equation}
This is the same as equation (13) in \cite{KPZKPDoussal}, which was also previously derived (in a slightly different form) for the case $\sigma=\sigma_{\rm KPZ}$ in \cite{QuastelRemenik}.
The equation \eqref{cylindricalKdV} is deduced in \cite{KPZKPDoussal} using the general results of P\"oppe and Sattinger \cite{PoppeSattinger} relating the KP hierarchy with a particular class of Fredholm determinants which seem to play a pivotal role in the field of integrable probability.
As observed by Quastel and Remenik in \cite{QuastelRemenik}, one can use, indeed, the arguments developed in \cite{PoppeSattinger} to prove a particular case (the scalar one) of the general relationship connecting $n$-space point distribution functions of the KPZ fixed point with the matrix KP equation.
A similar situation appears in the recent preprint \cite{BaikLiuSilva} where the KP equation is related to the one-point distribution of periodic TASEP, together with other integrable equations, such as a system of coupled modified KdV equations and non-linear heat equations.
All these equations are deduced in \cite{BaikLiuSilva} using RH techniques, but in the case of KP, the theory developed in \cite{PoppeSattinger} applies in a straightforward way.
The results of \cite{PoppeSattinger} are also used in \cite{LNR} to relate the distribution of the supremum of the Airy process with $m$ wanderers to the KdV equation.
For another approach relating Fredholm determinants and the KP equation, see \cite{AdlerShiotavanMoerbeke}.
\end{remark}

Let us now look at the KdV solutions $u_\sigma\left(x,t\right)$ in the $t\to 0$ limit.
It turns out that the initial data are not well-defined for $x<0$, more precisely $u_\sigma\left(x,t\right)$ will behave like $\frac{x}{2t}$ as $x<0$, $t\to 0$. 
For $x>0$, the initial data are well-defined.
In our next result, we present uniform in $x\in\left(-\infty,K\right]$ asymptotics for $u_\sigma\left(x,t\right)$ as $t\to 0$, for any $K>0$.
The function $v_\s$ describing the initial data for $x>0$ and appearing in {\it (iii)} of the theorem below is uniquely defined in terms of the solution to a model RH problem, introduced in Section \ref{section:RH3} and proven to be solvable by a vanishing lemma in Appendix \ref{appvanishinglemma}.

\begin{theorem}
\label{thminitialvalue}
Suppose that $\sigma$ satisfies Assumptions \ref{assumptions} with $\gamma\in \left(0,1\right]$ and that there exist $c_1, c_2,c_3>0$ and $C>0$ such that
\begin{align}
\label{eq:sigmaexpdecay}
&\left|\s(z)-\gamma \chi_{\left(0,+\infty\right)}(z)\right|\leq c_1{\rm e}^{-c_2|z|} &&\mbox{for all $z\in\mathbb R$,}
\\
\label{eq:sigmaprimedecay}
&\left|\sigma'(z)\right|\leq c_3|z|^{-2}&&\mbox{for $z\in\mathbb R$ such that }|z|>C.
\end{align}
\begin{itemize}
\item[(i)] For any $t_0>0$, there exist $M, c>0$ such that we have uniformly for $x\leq-M t^{1/3}$ and for $0<t<t_0$ that
\be
\label{estimatei}
u_\sigma\left(x,t\right)=\frac{x}{2t}+\O\left({\rm e}^{-c\frac{|x|}{t^{1/3}}}\right).
\ee
\item[(ii)] There exists $\epsilon>0$ such that for any $M>0$, we have uniformly for $|x|\leq M t^{1/3}$ and for $0<t<\epsilon$ that 
\be
u_\sigma\left(x,t\right)=\frac{x}{2t}-t^{-2/3}y_\g^2\left(-xt^{-1/3}\right)+{\O\left(1\right)},
\ee
with $y_\g$ as in \eqref{eq:PII}.
\item[(iii)] If $\g=1$, there exist $\epsilon, M>0$ such that for any $K>0$, we have uniformly for $M t^{1/3}\leq x\leq K$ and for $0<t<\epsilon$ that
\be
\label{eq:thmasymptoticsiii}
u_\sigma\left(x,t\right)=v_\s\left(x\right)\left(1+\O\left(x^{-1}t^{1/3}\right)\right),
\ee
where $v_\s$ is a function of $x>0$, independent of $t$.
\end{itemize}
\end{theorem}

\begin{figure}[htbp]
\centering
\hspace{1cm}
\begin{tikzpicture}[scale=1.3]

\draw[->] (-4,0) -- (5.5,0) node[right] {$x$};
\draw[-] (0,0) -- (0,0.8);
\draw[->](0,1.5) -- (0,4) node[above] {$t$};
\draw[domain=0:2.7, smooth, variable=\x, dashed] plot ({\x}, {.2*\x*\x*\x}) node[above] {${\scriptstyle x=Mt^{1/3}}$};
\draw[domain=0:2.7, smooth, variable=\x, dashed] plot ({-\x}, {.2*\x*\x*\x}) node[above] {${\scriptstyle x=-Mt^{1/3}}$};
\draw[-, dashed] (-2.54,3.3) -- (-4,3.3) node[left] {${\scriptstyle t=t_0}$};
\draw[-, dashed] (5,0) -- (5,3.7) node[above] {${\scriptstyle x=K}$};
\draw[-, dashed] (-1.98,1.55) -- (5,1.55) node[right] {${\scriptstyle t=\epsilon}$};

\node at (-3.4,1.7) {${ \frac{x}{2t}+\O\left({\rm e}^{-c\frac{|x|}{t^{1/3}}}\right)}$};
\node at (-3.4,2.2) {singular};
\node at (0,1) {${\frac{x}{2t}-\frac 1{t^{2/3}}y_\g^2\left(-\frac x{t^{1/3}}\right)+\O(1)}$};
\node at (0,1.35) {Painlev\'e II};
\node at (3.3,.7) {${v_\s\left(x\right)\left(1+\O\left(x^{-1}t^{1/3}\right)\right)}$};
\node at (3.3,1.2) {$\begin{array}{c}\mbox{integro-differential}\\\mbox{Painlev\'e V}\end{array}$};

\end{tikzpicture}
\caption{Phase diagram showing the different types of small $t$ asymptotic behavior for $u_\sigma(x,t)$.}
\label{figasymptotics}
\end{figure}
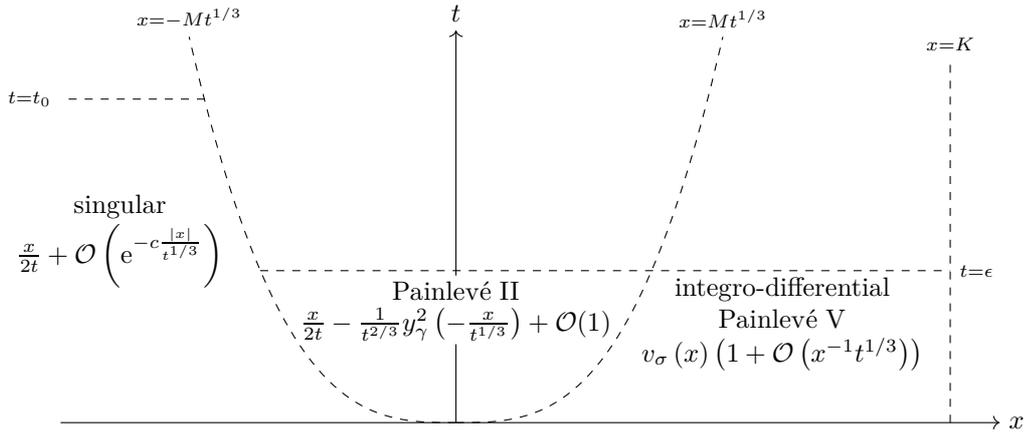

\begin{remark}\label{remark:gamma1}
The first two parts of the above result hold for any $\gamma\in(0,1]$, but the third part does not. We have no simple explanation for this, but it will turn out that our asymptotic analysis of the associated RH problem works only for $\gamma=1$.
\end{remark}

\begin{remark}
In view of Remark \ref{remark:cylKdV}, we observe that for $x<0$, although the initial data of $u_\sigma\left(x,t\right)$ are not well-defined, the initial data of $U_\sigma\left(\rho,\tau\right)$, i.e.\ the solution to the cylindrical KdV equation, are. On the other hand, for $x>0$, the initial data for KdV are well-defined but the initial data for cylindrical KdV are not. 

Scattering and inverse scattering theory for KdV solutions decaying at $\pm\infty$ is classical; a similar theory also exists for certain unbounded KdV solutions \cite{DubrovinMinakov} and for decaying cylindrical KdV solutions \cite{ItsSukhanov}. However, the solutions under consideration in the present work do not fit in these classes of solutions.
\end{remark}

\begin{remark}
Thanks to the estimate \eqref{estimatei} we can integrate twice the relation $\pa_x^2\log Q_\s(x,t)=u_\s(x,t)-\frac x{2t}$ to write
\begin{equation}
Q_\s(x,t)=\exp\left[\int_{-\infty}^x(x-\xi)\left(u_\s(\xi,t)-\frac \xi{2t}\right)\d\xi\right].
\end{equation}
\end{remark}

In the next result we gather some properties of the function $v_\s$ appearing in {\it (iii)} of Theorem \ref{thminitialvalue}.

\begin{theorem}
\label{thmv}
Suppose that $\sigma$ satisfies Assumptions \ref{assumptions} with $\gamma=1$ and that there exist $c_1, c_2,c_3>0$ and $C>0$ such that \eqref{eq:sigmaexpdecay} and \eqref{eq:sigmaprimedecay} hold true. Then the function $v_\s$ appearing in \eqref{eq:thmasymptoticsiii} of Theorem \ref{thminitialvalue} has the asymptotics
\begin{equation}
\label{as:v}
v_\s\left(x\right)=\frac{1}{8x^2}+\frac 12\int_\R\left(\chi_{(0,+\infty)}\left(r\right)-\s\left(r\right)\right)\d r+\mathcal O\left(x^2\right),\qquad \mbox{as $x\to 0$.}
\end{equation}
Moreover, it can be expressed for all $x>0$ as 
\be\label{eq:v1}
v_\s\left(x\right)=\frac 1x \int_{\mathbb R}\left( r\psi_\sigma^2\left(x;r\right)+\left(\pa_x\psi_\s\left(x;r\right)\right)^2\right) \d\s(r),
\ee
where $\psi_\sigma\left(x;z\right)$ solves the Schr\"odinger equation with potential $2v_\sigma$,
\be\label{eq:psiSchrodinger}
\partial_x^2\psi_\sigma\left(x;z\right)=\left(z - 2v_\sigma\left(x\right)\right)\psi_\sigma\left(x;z\right),\qquad z\in\R,\ x>0,
\ee
satisfies
\be\label{eq:psiintegral}
\int_\R\psi_\sigma^2\left(x;r\right)\d\s(r)=\frac x2,
\ee
and has the asymptotic behavior
\begin{equation}\label{eq:psibessel}
\psi_\sigma\left(x;z\right)\sim\sqrt{\frac x2}\,{\rm I}_0\left(x\sqrt z\right),\qquad\mbox{as $z\to \infty$ with $|\arg z|<\pi-\delta$, for any $\delta>0$,}
\end{equation}
where ${\rm I}_0$ is the modified Bessel function of the first kind.
\end{theorem}

\begin{remark}
We have no evidence that the above conditions determine $v_\s$ uniquely.
As already mentioned, we will give a unique characterization of $v_\s$ in terms of a RH problem later.
The equations \eqref{eq:v1}--\eqref{eq:psiintegral} yield the system
\be
\label{idpv}
\begin{cases}
x\partial_x^2\psi_\sigma\left(x;z\right)-xz\psi_\s\left(x,z\right)+ 2\psi_\sigma\left(x;z\right)\displaystyle\int_\R\left( r\psi_\sigma^2\left(x;r\right)+\left(\pa_x\psi_\s\left(x;r\right)\right)^2\right) \d\s(r)=0,
\\
\displaystyle\int_\R\psi_\sigma^2\left(x;r\right)\d\s(r)=\frac x2,
\end{cases}
\ee
which can be regarded as an integro-differential generalization of the Painlev\'e V equation. Indeed, when $\s$ is the piecewise constant function
$\s(r)=\mu\chi_{[-1,0]}+\chi_{\left(0,+\infty\right)}$ with
$0<\mu<1$, it can be shown that \eqref{idpv} reduces to the (special case of the) Painlev\'e V equation
\begin{equation*}
x q\left(x\right) \left(1 - q^2\left(x\right)\right)\left(x q\left(x\right) q'\left(x\right)\right)' + 
 x \left(1 - q^2\left(x\right)\right)^2 \left(\left(x q'\left(x\right)\right)' + \frac {1}4 q\left(x\right) \right)+ 
 x^2 q\left(x\right)\left(q\left(x\right) q'\left(x\right)\right)^2=0,
\end{equation*}
where we denote $'=\pa_x$ for brevity, for the quantity $q\left(x\right)$ defined by $\psi\left(x;-1\right)=\sqrt{\frac x{2\mu}}q\left(x^2\right)$.
When $\s$ is a more general piecewise constant increasing function with $\s\left(r\right)=1$ for all $r>0$, the system \eqref{idpv} turns into the system of coupled Painlev\'e V equations from \cite{CharlierDoeraene}.
\end{remark}

Finally, we also obtain small $t$ asymptotics for the Fredholm determinants $Q_\sigma\left(x,t\right)$.

\begin{theorem}
\label{thmQ}
Suppose that $\sigma$ satisfies Assumptions \ref{assumptions} with $\gamma\in \left(0,1\right]$ and that there exist $c_1, c_2,c_3>0$ and $C>0$ such that \eqref{eq:sigmaexpdecay} and \eqref{eq:sigmaprimedecay} hold true.
\begin{itemize}
\item[(i)] For any $t_0>0$, there exist $M, c>0$ such that we have uniformly for $x\leq -M t^{1/3}$ and for $0<t<t_0$ that
\be
\log Q_\sigma\left(x,t\right)=\O\left(e^{-c\frac{|x|}{t^{1/3}}}\right).
\ee
\item[(ii)] There exists $\epsilon>0$ such that for any $M>0$, we have uniformly for $|x|\leq M t^{1/3}$ and for $0<t<\epsilon$ that
\be
\log Q_\sigma\left(x,t\right)=\log F_{\rm TW}\left(-xt^{-1/3};\g\right)+\mathcal O\left(t^{1/3}\right),
\ee
where $F_{\rm TW}$ is the Tracy-Widom distribution, see Examples \ref{examples}.
\item[(iii)] If $\g=1$, there exist $\epsilon, M>0$ such that for any $K>0$, we have uniformly for $M t^{1/3}\leq x\leq K$ and for $0<t<\epsilon$ that
\be
\log Q_\sigma\left(x,t\right)=-\frac{x^3}{12t}-\frac{1}{8}\log (xt^{-1/3}) +\frac{\log 2}{24} +\log \zeta'(-1)
+\int_0^x (x-\xi) \left(v_\s(\xi)-\frac{1}{8\xi^2}\right)\d\xi+\mathcal O(x^{-1}t^{1/3}).
\ee
\end{itemize}
\end{theorem}

\begin{remark}\label{remark:sym}
If $\s(r)$ satisfies $\s(-r)=1-\s(r)$ (in particular $\g=1$) and if $\s'(r)$ decays sufficiently fast as $r \to \pm \infty$ (it is enough to strengthen \eqref{eq:sigmaprimedecay} to $\left|\sigma'(z)\right|\leq c_3|z|^{-3}$ for $|z|>C$) the accuracy in part {\it (ii)} of Theorems \ref{thminitialvalue} and \ref{thmQ} is improved to
\begin{align*}
u_\sigma\left(x,t\right)&=\frac{x}{2t}-t^{-2/3}y_1^2\left(-xt^{-1/3}\right)+\O\left(t^{2/3}\right),
\\
\log Q_\sigma\left(x,t\right)&=\log F_{\rm TW}\left(-xt^{-1/3};\g\right)+\mathcal O\left(t\right),
\end{align*}
respectively, see the proof of Lemma \ref{lemma:RH2jump} and Remark \ref{remsigmaKPZantisymmetry}.
This is the case for instance when $\s=\s_{\rm KPZ}$ in \eqref{sigmaKPZ}.
\end{remark}

\begin{remark}
Recall that $Q_{\sigma_{\rm KPZ}}(x,t)$ characterizes the narrow wedge solution $\mathcal H(X,T)$ of the KPZ equation \eqref{eq:KPZ}. In particular from \eqref{identityKPZ}, writing $\Upsilon_T=\frac{\mathcal H(2T,0)+\frac{T}{12}}{T^{1/3}}$, we have the identity
\[
Q_{\sigma_{\rm KPZ}}(x=sT^{-1/6},t=T^{-1/2})=	\mathbb E_{\mathrm{KPZ}}\left[{\rm e}^{-{\rm e}^{T^{1/3}(\Upsilon_T+s)}}\right],\]
where $\mathbb E_{\mathrm{KPZ}}$ is the expectation with respect to the random white noise term in the KPZ equation.
As a consequence, one can derive tail asymptotics for $\Upsilon_T$ from asymptotics for $Q_{\sigma_{\rm KPZ}}(sT^{-1/6},T^{-1/2})$.
Our results imply large $T$ asymptotics for $s\leq KT^{1/6}$, for any $K>0$, and refine the upper and lower bounds obtained in \cite{CorwinGhosal}.
We should also mention that asymptotics are known \cite{CafassoClaeys, Tsai} for $s\geq \epsilon T^{2/3}$ and $T$ large for any $\epsilon$, or equivalently $xt\geq \epsilon$ and $t$ small, and read
\[\log Q_{\sigma_{\rm KPZ}}(x,t)=-t^{-4}\phi(xt)-\frac{1}{6}\sqrt{1+\pi^2 xt} 
+\mathcal O(\log^2 x)+\mathcal O(t^{-2/3}),\]
where
\[
\phi(y) := \frac{4}{15\pi^6}\left(1+\pi^2y\right)^{5/2}-\frac{4}{15\pi^6}-\frac{2}{3\pi^4}y-\frac{1}{2\pi^2}y^2.
\]
Expanding $\phi(y)$ for small $y$, we recover as a leading order term $-\frac{x^3}{12 t}$, which matches with the asymptotics we obtain in case {\it (iii)} above. Despite this formal matching of the leading order terms, we should emphasize that
precise asymptotics for $Q_{\sigma_{\rm KPZ}}(x,t)$ are currently not known as $x\to\infty$, $xt\to 0$. The best available rigorous results in this regime are upper and lower bounds obtained in \cite{CorwinGhosal} (see also \cite{KrajenbrinkLeDoussal, SasorovMeersonProlhac} for related results). We plan to come back to the study of the different asymptotic regimes for $\Upsilon_T$ and their relations with Theorem \ref{thmQ} in a subsequent paper.
\end{remark}

\paragraph{Methodology.}
For the proof of Theorem \ref{theorem:main}, we rely on the Its-Izergin-Korepin-Slavnov method \cite{IIKS,BC} to relate the determinants $Q_\sigma$ to a RH problem, and on Lax pair techniques to relate the relevant RH problem with differential equations.
A first unusual feature in our methodology is that we cannot deform the jump contour in the RH problem away from the real line (which was done for instance in \cite{CafassoClaeys} in the case $\sigma=\sigma_{\rm KPZ}$), because we do not require analyticity of the function $\sigma$.
As a consequence, we cannot prove directly that the RH solution is differentiable with respect to various parameters, which we need in order to apply Lax pair techniques. We can resolve this issue by using the operator theory underlying the theory of RH problems, see Section \ref{sec:3}.
A second key element in our approach is the derivation of the identity \eqref{eq:u2}, which is an integral expression for the KdV solution in terms of the Schr\"odinger eigenfunction; the derivation of this identity is a direct residue computation based on the RH problem, see Section \ref{sec:4}. This identity allows us to derive the integro-differential Painlev\'e II equation by direct substitution into the Schr\"odinger equation; in particular, a feature of this derivation is that we do not obtain the integro-differential Painlev\'e II equation from a compatibility condition, like in \cite{Bothner}, but rather from an eigenvalue equation. One more derivation of the integro-differential Painlev\'e II equation, which can be extended to the case of the point processes related to multi-critical fermions, is given in \cite{Krajenbrink}. 

For the proofs of Theorem \ref{thminitialvalue}
and Theorem \ref{thmQ}, we use a Deift-Zhou nonlinear steepest descent \cite{DZ} asymptotic analysis of the RH problem. Whereas this analysis is standard in the case of negative $x$, we encounter several obstacles for a direct application of the method when $x$ is small and even more when $x$ is positive. 
For $x$ small, the most important obstacle is once more that, because $\sigma$ is not necessarily analytic, we cannot deform the jump contours for the RH problems, which is usually an important step in the Deift-Zhou steepest descent method, often called opening of lenses. However, in our setting, we can adapt the RH analysis to avoid this step, see in particular Section \ref{sec:proofii}. 
For $x$ positive, we encounter a similar problem, but in addition to that, for the construction of a local parametrix we need a model RH problem whose solution we cannot construct explicitly. We emphasize once more that we construct this local parametrix without being able to deform jump contours away from the real line, see Section \ref{section:RH3}.

\paragraph{Outline for the rest of the paper.}
In Section \ref{sec:2}, we will relate the determinants  $Q_\sigma$ to a RH problem. 
In Section \ref{sec:3}, we will associate a Lax pair 
to the RH problem, which will allow us to relate it to the KdV equation. In Section \ref{sec:4}, we will further exploit the Lax pair and the RH representation to derive the integro-differential Painlev\'e II equation and to complete the proof of Theorem \ref{theorem:main}.
In Section \ref{sec:5}, we will establish the small time asymptotics of the KdV solution $u_\sigma(x,t)$ for $x\leq -Mt^{1/3}$, $M>0$; in Section \ref{sec:proofii}, we will obtain small time asymptotics for $u_\sigma(x,t)$ for $-Mt^{1/3}\leq x\leq Mt^{1/3}$; in Section \ref{section:RH3} finally, we will derive small time asymptotics for $u_\sigma(x,t)$ when $Mt^{1/3}\leq x\leq K$, and thus complete the proofs of Theorem \ref{thminitialvalue} and Theorem \ref{thmQ}. This will require the use of a model RH problem, whose solvability we will prove in Appendix \ref{appvanishinglemma}, and which encodes information about the function $v_\sigma$ appearing in Theorem \ref{thminitialvalue}. We will also prove the properties of $v_\sigma$ described in Theorem \ref{thmv} in Section \ref{section:RH3}.

\section{RH characterization of $Q_\sigma\left(x,t\right)$}\label{sec:2}

In this section, we will express $\partial_x \log Q_\sigma\left(x,t\right)$, with $Q_\sigma\left(x,t\right)$ defined as in \eqref{def:Q}, in terms of the solution of a $2 \times 2$ RH problem. We will suppose here that $\sigma$ is such that Assumptions \ref{assumptions} hold for some $\gamma\in\left(0,1\right]$.
We follow a procedure developed by Its, Izergin, Korepin, and Slavnov \cite{IIKS}.

We define for any $x\in\mathbb R$, $t>0$, functions $\mathbf{f}\left(u\right)=\mathbf{f}\left(u;x,t\right)$ and $\mathbf{h}\left(v\right)=\mathbf{h}\left(v;x,t\right)$ as follows,
\begin{equation}
\label{def:fgAiry}
	\mathbf{f}\left(u\right) :=\sqrt{\sigma\left(t^{-2/3}u+x/t\right)} \begin{pmatrix} -\i \Ai'\left(u\right) \\ \Ai\left(u\right) \end{pmatrix}, \quad \mathbf{h}\left(v\right) := \sqrt{\sigma\left(t^{-2/3}v+x/t\right)}\begin{pmatrix} -\i\Ai\left(v\right) \\ \Ai'\left(v\right) \end{pmatrix},\qquad u,v\in\mathbb R.
\end{equation}
We can re-write the kernel $K_{\sigma,x,t}$ in \eqref{def:kernel} as
\begin{equation*}
K_{\sigma,x,t}\left(u,v\right)=\frac{\mathbf{f}^\top\left(u\right)\mathbf{h}\left(v\right)}{u-v}
\end{equation*}
which means that the corresponding integral operator $\K^\Ai_{\sigma,x,t}$ is integrable, in the sense of \cite{IIKS}.
Recall from Remark \ref{remark:traceclassandpositivity} that $Q_\sigma\left(x,t\right)=\det\left(1-\K^\Ai_{\sigma, x, t}\right)\in \left(0,1\right]$, so $1 - \mathbb K^\Ai_{\sigma, x, t}$ is invertible and the resolvent operator
$$\L^\Ai_{\sigma, x, t} := \left(1 - \K^\Ai_{\sigma, x, t}\right)^{-1}\K^\Ai_{\sigma, x, t}$$ 
exists.
We now define the $2\times 2$ matrix-valued function
\begin{equation*}
	Y\left(\zeta\right) := I - \int_{\mathbb R} \frac{\mathbf{F}\left(\xi\right)\mathbf{h}^{\top}\left(\xi\right)}{\xi - \zeta}\d\xi, \quad \text{with} \qquad \mathbf{F}\left(\xi\right) = \left[\left(1 - \K^\Ai_{\sigma, x, t}\right)^{-1} \mathbf{f} \right]\left(\xi\right),
\end{equation*}
where the operator $\left(1 - \K^\Ai_{\sigma,x,t}\right)^{-1}$ acts on the vector $\mathbf{f}$, defined in \eqref{def:fgAiry}, component-wise; note that the components of $\mathbf{f}$ are in $L^2$, by the asymptotic behavior of the Airy function and its derivative, and by Assumptions \ref{assumptions}.
We know from \cite{IIKS, DIZ} that the resolvent $\L^\Ai_{\sigma, x, t}$ is also an operator of integrable type, whose kernel is given by
\begin{equation}
\label{eq:kernelresolvent}
L^\Ai_{\sigma, x, t}\left(u,v\right) = \frac{1}{u - v} \mathbf{f}^{\top}\left(u\right)Y_+^{\top}\left(u\right)Y_+^{-\top}\left(v\right)\mathbf{h}\left(v\right),
\qquad u,v\in\R,
\end{equation}
and that $Y$ satisfies the following RH problem, for any $x\in\mathbb R, t>0$.
\subsubsection*{RH problem for $Y$}
\begin{enumerate}[label=(\alph*)]
\item $Y$ is analytic in $\mathbb C\setminus \mathbb R$.
\item $Y$ has boundary values $Y_\pm$ on $\R$ with $Y_\pm-I\in L^2\left(\mathbb R\right)$ and which are continuous except possibly at the points $\zeta_j:=t^{-2/3}r_j + x/t$, $j=1,\ldots, k$, and they are related by
\begin{equation}\label{jumpY1}
	Y_+\left(\zeta\right)=Y_-\left(\zeta\right)J_Y\left(\zeta\right), \quad J_Y\left(\zeta\right)=I-2\pi\i\, \mathbf{f}\left(\zeta\right)\mathbf{h}^{\top}\left(\zeta\right) 
\end{equation}
or more explicitly by
\begin{equation*}
	 Y_+\left(\zeta\right)=Y_-\left(\zeta\right)\begin{pmatrix} 
				1 + 2\pi\i \sigma\left(t^{-2/3}\zeta + x/t\right)\Ai\left(\zeta\right)\Ai'\left(\zeta\right) & - 2\pi  \sigma\left(t^{-2/3}\zeta + x/t\right)\Ai'^2\left(\zeta\right)\\
				-2\pi \sigma\left(t^{-2/3}\zeta + x/t\right)\Ai^2\left(\zeta\right) & 1 - 2\pi\i \sigma\left(t^{-2/3}\zeta + x/t\right)\Ai\left(\zeta\right)\Ai'\left(\zeta\right) 
				\end{pmatrix}.
\end{equation*}
\item As $\zeta\to\infty$, we have, for some $\alpha=\alpha\left(x,t\right)$, $\beta=\beta\left(x,t\right)$, $\eta=\eta\left(x,t\right)$ that
\begin{equation}
\label{asY}
Y\left(z\right)=I+\frac{Y_1}{\z}+\mathcal O\left(\z^{-2}\right),\qquad Y_1 := \begin{pmatrix}\beta&-\i\eta\\
\i\alpha&-\beta
\end{pmatrix}.
\end{equation}
\end{enumerate}

We now transform the RH problem for $Y$ in such a way that the Airy functions disappear from the jump matrices. For this, we use a well-known undressing procedure, which consists of defining 
\begin{equation}\label{eq:defPhi}
\Phi_{\Airy}\left(\zeta\right) =\begin{cases} -\sqrt{2 \pi }\begin{pmatrix} \Ai'\left(\zeta\right) &-\omega\Ai'\left(\omega^2 \zeta\right) \\ \i\Ai\left(\zeta\right) & -\i\omega^2\Ai\left(\omega^2 \zeta\right)  \end{pmatrix} &\mbox{for } \Im \zeta>0,\\ 
	  -\sqrt{2 \pi }\begin{pmatrix} \Ai'\left(\zeta\right) &\omega^2\Ai'\left(\omega \zeta\right) \\ \i\Ai\left(\zeta\right) & \i\omega\Ai\left(\omega \zeta\right)  \end{pmatrix}&\mbox{for } \Im \zeta<0,\end{cases}
\end{equation}
where $\omega=\exp\left(2\pi\i/3\right)$.
Using the asymptotic behavior of the Airy function and its derivative (see, for instance, \cite{DLMF}) and the algebraic relation
$$
	\Ai\left(\zeta\right) + \omega\Ai\left(\omega\zeta\right) + \omega^2\Ai\left(\omega^2\zeta\right) = 0,
$$
one can show that $\Phi_{\Airy}$ satisfies the following RH problem.

\subsubsection*{RH problem for $\Phi_{\Airy}$}
\begin{enumerate}[label=(\alph*)]
\item $\Phi_{\Airy}$ is analytic in $\mathbb C\setminus \mathbb R$.
\item $\Phi_{\Airy}$ has continuous boundary values $\Phi_{\sf\Airy,\pm}$ satisfying the jump condition
\begin{equation}
	\Phi_{\Airy,+}\left(\zeta\right)=\Phi_{\Airy,-}\left(\zeta\right)\begin{pmatrix}
								1 & 1 \\
								0 & 1
							\end{pmatrix}, \quad \zeta \in \mathbb R.
\end{equation}
\item As $\zeta\to\infty$, $\Phi_\Airy$ has the asymptotic behavior
\begin{multline*}
\Phi_\Airy\left(\zeta\right) = \left( I +\frac{1}{\zeta}\begin{pmatrix}0&\frac {7\i}{48}\\
0&0\end{pmatrix} +\mathcal O\left(\frac{1}{\zeta^2}\right) \right) \zeta^{\frac{1}{4} \sigma_3} A^{-1} {\rm e}^{-\frac{2}{3}\zeta^{3/2}\sigma_3}\times\begin{cases}I,&|\arg\zeta|<\pi-\delta,\\
\begin{pmatrix}1&0\\\mp 1&1\end{pmatrix},&\pi-\delta<\pm \arg \zeta<\pi,\end{cases}
\end{multline*}
for any $0<\delta<\pi/2$; here the principal branches of $\zeta^{\frac{1}4\sigma_3}$ and $\zeta^{1/2}$ are taken, analytic in $\mathbb C \backslash \left(-\infty,0\right]$ and positive for $\zeta>0$, and
 \[A = \left(I + \i \sigma_1\right) / \sqrt{2},\qquad \sigma_1=\begin{pmatrix} 0&1 \\ 1 & 0 \end{pmatrix},\qquad \sigma_3 = \begin{pmatrix} 1 & 0 \\ 0 & -1 \end{pmatrix}.\]
\end{enumerate}

Next, we define the affine transformation
\begin{equation}\label{zetaparametrization}
	\zeta\left(z\right) = t^{2/3}z - t^{-1/3}x,
\end{equation}
with inverse 
\begin{equation*}
	z\left(\zeta\right) = t^{-2/3}\zeta + x/t.
\end{equation*}
Finally, we define
\begin{equation}\label{eq:defPsi}
	\Psi\left(z\right) := \begin{pmatrix}1& \frac{\i x^2}{4t}\\0&1
	\end{pmatrix}t^{-\frac{1}{6} \sigma_3}Y\left(\zeta\left(z\right)\right)\Phi_{\Airy}\left(\zeta\left(z\right)\right).
\end{equation}
We now show, using the RH conditions for $\Phi_\Airy$ and $Y$, that $\Psi$ satisfies the following RH problem.

\subsubsection*{RH problem for $\Psi$}
\begin{enumerate}[label=(\alph*)]
\item $\Psi$ is analytic in $\mathbb C\setminus \mathbb R$.
\item $\Psi$ has boundary values $\Psi_\pm$ on $\R$ with $\Psi_\pm\in L^2_{\rm loc}(\R)$ and which are continuous on the real line except at the points $r_1, \ldots ,r_k$, and they are related by
\begin{equation}\label{jumpPsi}
\Psi_+\left(z\right)=\Psi_-\left(z\right)\begin{pmatrix}1&1-\s(z)\\0&1\end{pmatrix}.
\end{equation}
\item As $z\to\infty$, there exist functions $p= p_{\sigma}\left(x,t\right), q= q_{\sigma}\left(x,t\right)$ and $r= r_{\sigma}\left(x,t\right)$ such that $\Psi$ has the asymptotic behavior
\begin{multline}
\label{eq:Psiasymp}
\Psi\left(z\right) = \left(I + \frac{1}{z} \begin{pmatrix} q & -\i r \\ \i p & -  q \end{pmatrix} + \mathcal O\left(\frac{1}{z^2}\right) \right) z^{\frac{1}{4} \sigma_3} A^{-1} {\rm e}^{\left(-\frac{2}{3}tz^{3/2} + x z^{1/2}\right)\sigma_3}
\\		
\times\begin{cases}I,&|\arg z|<\pi-\delta,\\
\begin{pmatrix}1&0\\\mp 1&1\end{pmatrix},&\pi-\delta<\pm \arg z<\pi,\end{cases}
\end{multline}
for any $0<\delta<\pi/2$; here the principal branches of $z^{\frac{1}4\sigma_3}$ and $z^{1/2}$ are taken, analytic in $\mathbb C \setminus \left(-\infty,0\right]$ and positive for $z>0$.
\end{enumerate}

\begin{remark}
A similar procedure has been used in \cite{CafassoClaeys} in the case $\sigma = \sigma_{\mathrm{KPZ}}$, but with a different set of contours for $\Psi$ (consisting of four different rays) and uniform asymptotics at infinity. Here, we are forced to make the unusual choice of formulating the RH problem for $\Psi$ with a jump on the real axis only and with slightly more involved asymptotics at infinity, because we do not assume $\s$ to have an analytic extension to a neighborhood of the real line.
\end{remark}

Condition (a) above follows from the fact that both $Y$ and $\Phi_{\Airy}$ are analytic function on $\mathbb C\setminus \mathbb R$. The continuity and $L^2$ conditions in (b) follow from the analogous conditions for $Y$, and the jump relation follows from the relations
\begin{equation}
\label{usefulfhequations}
	\mathbf{f}\left(\zeta\left(z\right)\right) = \i\sqrt{\frac{\s(z)}{2 \pi}}\Phi_{\Airy,-}\left(\zeta\left(z\right)\right)\begin{pmatrix} 1 \\ 0 \end{pmatrix},
\quad 
\mathbf{h}^\top\left(\zeta\left(z\right)\right) = \sqrt{\frac{\s(z)}{2 \pi}} \begin{pmatrix} 0 & -1 \end{pmatrix}\Phi_{\Airy,+}^{-1}\left(\zeta\left(z\right)\right)
\end{equation}
that are to be used in the last step of the following chain of equalities,
\begin{align}
\nonumber
	\Psi_-^{-1}\left(z\right)\Psi_+\left(z\right)&=
	\Phi_{\Airy,-}^{-1}\left(\zeta\left(z\right)\right)Y_-^{-1}\left(\zeta\left(z\right)\right)Y_+\left(\zeta\left(z\right)\right)\Phi_{\Airy,+}\left(\zeta\left(z\right)\right) \\ 	
	&=  \Phi_{\Airy,-}^{-1}\left(\zeta\left(z\right)\right)\left(I - 2\pi\i\,\mathbf{f}\left(\zeta\left(z\right)\right)\mathbf{h}^\top\left(\zeta\left(z\right)\right)\right) \Phi_{\Airy,+}\left(\zeta\left(z\right)\right)=
	\begin{pmatrix}
	1 & 1-\sigma(z) \\ 0 & 1
	\end{pmatrix}, \label{proofjump}
\end{align}
valid for all $z \in \mathbb R$.
Condition (c) for $\Psi$ follows from the analogous condition for $\Phi_{\Airy}$, together with the asymptotics $Y\left(\zeta\right) =I+\mathcal O(1/\zeta)$ as $\zeta \to \infty$ and the change of variables \eqref{zetaparametrization}.
\qed \\

In what follows, we will need the following equations, relating the asymptotic expansions of $Y$ and $\Psi$ at infinity.

\begin{proposition}
	Let $p,q,r$ defined as in \eqref{eq:Psiasymp}, and $\alpha,\beta,\eta$ as in \eqref{asY}. Then 
\begin{align}
\nonumber
p\left(x,t\right)&=\frac{\alpha\left(x,t\right)}{t^{1/3}}+\frac{x^2}{4t},\qquad q\left(x,t\right)=\frac{t^{2/3}\beta\left(x,t\right)-x^2\alpha\left(x,t\right)}{t^{4/3}}-\frac{x^4+8tx}{32t^2},
\\
\label{eq:pqalphabeta}
r\left(x,t\right)&=-\frac{\alpha(x,t) x^4}{16t^{8/3}}+\frac{\beta(x,t) x^2(1+t^{1/3})}{4t^2}+\frac{\eta(x,t)} t-\frac{28 t^2+16 t x^3+x^6}{192 t^3}.
\end{align}
\end{proposition}
\proof

We observe that $\z^{\s_3/4}\left(z\right)A^{-1}\e^{-\frac 23\z^{3/2}\left(z\right)\s_3}$ behaves, for large $z$, as
\begin{equation*}
t^{\s_3/6}
\begin{pmatrix}1 & -\frac {\i x^2}{4t} \\ 0 & 1 \end{pmatrix}
\left(
	I+
	\frac 1z \renewcommand*{\arraystretch}{1.3}\begin{pmatrix}-\frac{x^4+8tx}{32t^2} & \frac{\i x^3 \left(16 t+x^3\right)}{192 t^3} \\ \frac{\i x^2}{4t} & \frac{x^4+8tx}{32t^2} \end{pmatrix}
	+\O\left(z^{-2}\right)
\right)
z^{\s_3/4}A^{-1}\e^{\left(-\frac 23tz^{3/2}+xz^{1/2}\right)\s_3}.
\end{equation*}
Hence
\begin{equation*}
\Psi\left(z\right)\sim
\left(I+\frac {\wt Y_1+\wt\Phi_{\Airy,1}}{t^{2/3}z}+\O\left(z^{-2}\right)\right)
\left(
	I+
	\frac 1z \renewcommand*{\arraystretch}{1.3}\begin{pmatrix}-\frac{x^4+8tx}{32t^2} & \frac{\i x^3 \left(16 t+x^3\right)}{192 t^3} \\ \frac{\i x^2}{4t} & \frac{x^4+8tx}{32t^2} \end{pmatrix}
	+\O\left(z^{-2}\right)
\right)
z^{\s_3/4}A^{-1}\e^{\left(-\frac 23tz^{3/2}+xz^{1/2}\right)\s_3}
\end{equation*}
where
\begin{equation*}
\wt Y_1=\begin{pmatrix}1 & \frac {\i x^2}{4t} \\ 0 & 1 \end{pmatrix}t^{-\s_3/6}Y_1t^{\s_3/6}\begin{pmatrix}1 & -\frac {\i x^2}{4t} \\ 0 & 1 \end{pmatrix},\qquad 
\wt\Phi_{\Airy,1}=\begin{pmatrix}1 & \frac {\i x^2}{4t} \\ 0 & 1 \end{pmatrix}t^{-\s_3/6}\begin{pmatrix} 0 & \frac {7\i}{48} \\ 0 & 0\end{pmatrix}t^{\s_3/6}\begin{pmatrix}1 & -\frac {\i x^2}{4t} \\ 0 & 1 \end{pmatrix},
\end{equation*}
and the proof is completed through straightforward computations.
\qed\\

We conclude the section with the derivation of the following relation between the Fredholm determinant $Q_{\sigma}$ and $\Psi$.

\begin{proposition}\label{prop:PsiQ}
	The RH problem for $\Psi$ has a unique solution $\Psi = \Psi\left(z;x,t\right)$ for any $x \in \R$ and $t > 0$, and we have the identity
	\begin{equation}\label{eq:PsiQ}
		\partial_x \log Q_\sigma\left(x,t\right) = -\frac{1}{2\pi\i  t}\int_{\mathbb R} \left( \Psi_+^{-1}\left(r;x,t\right)\Psi'_+\left(r;x,t\right)
		\right)_{21}\d\s(r),
	\end{equation}
denoting $(')=\frac{\d}{\d r}$.
\end{proposition}
\begin{remark}\label{remark:continuityPsirk}
Although $\Psi\left(z;x,t\right)$ is not continuous at the points $z=r_k$, it is not hard to see (using the RH conditions for $\Psi$) that $\left( \Psi_+^{-1}\left(z;x,t\right)\Psi'_+\left(z;x,t\right)\right)_{21}$ exists and is continuous at $z=r_k$.
\end{remark}
\proof
The uniqueness of $\Psi$ follows from standard arguments, see e.g. \cite[Theorem 7.18]{Deift}.
Indeed, let $\Psi_1,\Psi_2$ be two solutions to the RH problem for $\Psi$ and set $T(z)=\Psi_1(z)\Psi_2^{-1}(z)$.
Along $\R$ we have $T_+=T_-$ and therefore $T(z)$ is a meromorphic function of $z$ with isolated singularities at $z=r_1,...,r_k$; however since the boundary values of $\Psi_{1,2}$ at $\R$ are locally $L^2$, it follows that each $r_j$ is a removable singularity, and that $T$ is entire.
Finally, condition (c) of the RH problem implies that $T(z)=I+\O\left(z^{-1}\right)$ as $z\to\infty$ so by Liouville's theorem $T=I$ identically.
(One can reason similarly for the uniqueness of $Y$.)
As for \eqref{eq:PsiQ}, we start by applying Jacobi variational formula
\begin{equation}
\label{eq:derlog}
	\partial_x \log Q_\s\left(x,t\right)=-\Tr\left[\left(1-\K^\Ai_{\sigma,x,t}\right)^{-1}\partial_x \K^\Ai_{\sigma,x,t}\right]=-t^{-1/3}\int_{\mathbb R}\frac{L_{\sigma,x,t}\left(\zeta\left(r\right),\zeta\left(r\right)\right)}{\s(r)}\d\s(r),
\end{equation}
where the trace is expressed by standard manipulations (see e.g. \cite[Exercise 9.3.1]{Forrester} and \cite{ClaeysDoeraene}) in terms of the resolvent kernel $L^\Ai_{\sigma,x,t}\left(u,v\right)$, given in \eqref{eq:kernelresolvent}.
We remark that in this formula, $\d\s(r)$ contains Dirac delta measures at the singularities of $\s$ (see the discussion in Section 1, below Theorem \ref{theorem:main}); however, as a direct consequence of \eqref{eq:kernelresolvent}, bearing in mind the definition \eqref{def:fgAiry} and the RH conditions satisfied by $Y$, the ratio of $L_{\sigma,x,t}\left(\zeta\left(r\right),\zeta\left(r\right)\right)$ by $\s(r)$ is continuous on the real line and so the integral with respect to $\d\s$ in \eqref{eq:derlog} is well defined.
Using now the expression 
\be
\nonumber
L^\Ai_{\sigma,x,t}\left(\zeta,\zeta\right) = \frac{\d}{\d \zeta} \left( \mathbf{f}^\top\left(\zeta\right) Y_+^\top\left(\zeta\right) \right) Y_+^{-\top}\left(\zeta\right) \mathbf{h}\left(\zeta\right),
\ee
stemming from \eqref{eq:kernelresolvent}, we obtain
\begin{multline}\label{eq1prfPsiQ}
	\partial_x \log Q_{\sigma}\left(x,t\right) = -t^{-1/3}\ \int_{\mathbb R}  \left(\Ai'^2\left(\zeta\left(r\right)\right) - \zeta\left(r\right) \Ai^2\left(\zeta\left(r\right)\right)\right)\d\s(r) \\-t^{-1/3}\ \int_\mathbb R \mathbf{h}^{\top}\left(\zeta\left(r\right)\right)Y_+^{-1}\left(\zeta\left(r\right)\right)Y_+'\left(\zeta\left(r\right)\right)\mathbf{f}\left(\zeta\left(r\right)\right) \frac{\d\s(r)}{\s(r)}.
\end{multline}
We now analyze the second term on the right hand side. Using \eqref{eq:defPsi} we obtain
\begin{multline*}
\int_\mathbb R \mathbf{h}^\top\left(\zeta\left(r\right)\right)Y_+^{-1}\left(\zeta\left(r\right)\right)Y_+'\left(\zeta\left(r\right)\right)\mathbf{f}\left(\zeta\left(r\right)\right) \frac{\d\s(r)}{\s(r)} \\
= \int_\mathbb R \mathbf{h}^\top\left(\zeta\left(r\right)\right) \Phi_{\Airy,+}\left(\zeta\left(r\right)\right)\Psi^{-1}\left(r\right) \frac{\d}{\d r}\left(\Psi\left(r\right)\Phi_{\Airy,+}^{-1}\left(\zeta\left(r\right)\right)\right)\mathbf{f}\left(\zeta\left(r\right)\right)\frac{\d\s(r)}{\s(r)},
\end{multline*}
which is equal to
\begin{multline}
 -\int_{\mathbb R} \mathbf{h}^\top\left(\zeta\left(r\right)\right) \Phi_{\Airy,+}'\left(\zeta\left(r\right)\right)\Phi_{\Airy,+}^{-1}\left(\zeta\left(r\right)\right) \mathbf{f}\left(\zeta\left(r\right)\right) \frac{\d\s(r)}{\s(r)}\\  + \int_{\mathbb R} \mathbf{h}^\top\left(\zeta\left(r\right)\right)\Phi_{\Airy,+}\left(\zeta\left(r\right)\right)\Psi^{-1}\left(r\right) \Psi'\left(r\right) \Phi_{\Airy,+}^{-1}\left(\zeta\left(r\right)\right) \mathbf{f}\left(\zeta\left(r\right)\right)\frac{\d\s(r)}{\s(r)}. \label{eq2prfPsiQ}  
\end{multline}
On the other hand, using \eqref{usefulfhequations} and the additional equalities
\begin{align*}
	\mathbf{f}\left(\zeta\left(z\right)\right) &= \sqrt{\frac{\s(z)}{2 \pi}} \Phi_{\Airy,+}\left(\zeta\left(z\right)\right) \begin{pmatrix} 1 \\ 0 \end{pmatrix}, \\
	 \mathbf{h}^\top\left(\zeta\left(z\right)\right) \Phi'_{\Airy,+}\left(\zeta\left(z\right)\right) &= -\sqrt{2\pi \sigma(z)} \begin{pmatrix} -\i \zeta\left(z\right)\Ai^2\left(\zeta\left(z\right)\right) + \i\Ai'^2\left(\zeta\left(z\right)\right) \\ \star \end{pmatrix}, 
\end{align*}
(where $\star$ is some value which does not play any role in the sequel) we can re-write \eqref{eq2prfPsiQ} as 
\begin{multline}
	\int_\mathbb R \mathbf{h}^\top\left(\zeta\left(r\right)\right)Y_+^{-1}\left(\zeta\left(r\right)\right)Y_+'\left(\zeta\left(r\right)\right)\mathbf{f}\left(\zeta\left(r\right)\right) \frac{\d\s(r)}{\s(r)}  \\ 
	=	\int_\mathbb R \left(\zeta
	{\left(r\right)}
\Ai^2\left(\zeta\left(r\right)\right) - \Ai'^2\left(\zeta\left(r\right)\right) \right)\d\s(r) + \frac{t^{-2/3}}{2\pi\i } \int_{\mathbb R} \left(\Psi^{-1}\left(r\right) \Psi'\left(r\right) \right)_{21} \d\s(r). \label{eq3prfPsiQ}  
\end{multline}
Note that, in the expression above, the prime denotes the derivative with respect to $\zeta$ in the first term and with respect to $z=r$ in the second.
Combining equation \eqref{eq1prfPsiQ} with  \eqref{eq3prfPsiQ} we finally obtain the proof of equation \eqref{eq:PsiQ}. \qed

\section{Derivation of the Lax pair for KdV}\label{sec:3}

In this section we want to associate a Lax pair to the RH problem for $\Psi$.
 More precisely, we want to find  linear matrix differential equations for $\Psi$ with respect to the variables $x$ and $t$. In fact, it will turn out convenient to slightly modify $\Psi$, as in \eqref{eq:defTheta} below, in such a way that the differential equations will become simpler.
Since the jump matrix for $\Psi$ does not depend on $x$ or $t$, we can use standard arguments to deduce such differential equations. For this,  we first need RH conditions for $\partial_x\Psi$ and $\partial_t\Psi$. We should note that it is easy to deduce RH conditions for the derivatives formally, but that it is less obvious how to justify the existence of the $x$- and $t$-derivatives of $\Psi$ (note that $\sigma$ is not necessarily analytic), and how to justify that the asymptotics for $\partial_x\Psi$ or $\partial_t\Psi$ can be obtained by formally differentiating \eqref{eq:Psiasymp}. 
{Using the operator theory and integral equations associated with RH problems, we will first show how to differentiate $Y$ with respect to $x$ and $t$, and then we will carry this over to $\Psi$.} Let us denote with $C$ the Cauchy operator defined by
\begin{equation}
	\left(Cg\right)\left(\zeta\right) = \frac{1}{2\pi\i } \int_\mathbb R g\left(\xi\right) \frac{\d \xi}{\xi - \zeta}, \quad \zeta \in \mathbb C\backslash \mathbb R,
\end{equation}
and with $C_{\pm}$ its non-tangential boundary values, as $\zeta$ approaches the real line from above ($+$) or from below ($-$).
Given $\zeta(z)$ as in \eqref{zetaparametrization} and the matrix-valued function $J_Y$ defined on $\mathbb R$ by \eqref{jumpY1}, we define
\begin{equation}\label{def:Yhat}
\widehat Y(z)=Y(\zeta(z)),\qquad J_{\widehat Y}(z)=J_Y(\zeta(z)),
\end{equation} 
such that $\widehat Y$ satisfies the following RH problem (equivalent to the one for $Y$).
\subsubsection*{RH problem for $\widehat Y$}
\begin{enumerate}[label=(\alph*)]
\item $\widehat Y$ is analytic in $\mathbb C\setminus \mathbb R$,
\item $\widehat Y$ has boundary values $\widehat Y_\pm$ on the real line with $\widehat Y_\pm -I \in L^2\left(\mathbb R\right)$ and which are continuous  except possibly at the points $r_1,\ldots, r_k$, and they are related by
\begin{equation}
	\widehat Y_+(z)=\widehat Y_-(z){J_{\widehat Y}(z)},\qquad z\in\mathbb R\setminus\left\lbrace r_1,\dots,r_k\right\rbrace,
\end{equation}
\item $ \widehat Y(z) = I+\dfrac{Y_1}{t^{2/3}z} + \mathcal O(z^{-2})$ as $z \to \infty$, with $Y_1$ as in \eqref{asY}.
\end{enumerate}
Next, we set
\begin{equation}
\label{def:Cauchyop}
C_{J_{\widehat Y}}\left(g\right) := C_-\left(g\left(J_{\widehat Y} - I\right)\right).
\end{equation}

Since $C_\pm$ are bounded operators on $L^2\left(\mathbb R\right)$, $C_{J_{\widehat Y}}$ is a bounded operator on $L^2\left(\mathbb R\right)$, because $J_{\widehat Y}-I\in L^\infty\left(\mathbb R\right)$ (this follows from the third part of Assumptions \ref{assumptions}). 
\begin{lemma}\label{lemma:YCauchy}
	Let $J_{\widehat Y}$ be as in \eqref{jumpY1}, \eqref{def:Yhat}. The operator $1 - C_{J_{\widehat Y}}$ is invertible on $L^2(\mathbb R)$ and the (unique) solution of the RH problem for $\widehat Y$ is given by 	\begin{equation}\label{eq:Cauchy1}
		\widehat Y = I + C\left(\widehat Y_-\left(J_{\widehat Y} - I\right)\right),
	\end{equation}
	where $\widehat Y_-$ is uniquely determined by the equation
	\begin{equation}\label{eq:Cauchy2}
		\widehat Y_- = I + C_{J_{\widehat Y}}\left(\widehat Y_-\right)
	\end{equation}
	or, equivalently,
	\begin{equation}\label{eq:Cauchy3}
		 \widehat Y_- - I = \left(1 - C_{J_{\widehat Y}}\right)^{-1}C_-\left(J_{\widehat Y} - I\right).
	\end{equation}
\end{lemma}
\begin{proof}
	This result is a rephrasing of the discussion of Section 2 in \cite{DIZ}. In particular, the fact that the operator $1 - C_{J_{\widehat Y}}$ is invertible is proven in Lemma 2.12 of \cite{DIZ}, and it follows from the fact that the operator $1 - \mathbb K_{\sigma,x,t}^\Ai$ is invertible, since $\det\left(1 - \mathbb K^\Ai_{\sigma,x,t}\right) \in \left(0,1\right]$. As for the equations \eqref{eq:Cauchy1} and \eqref{eq:Cauchy2}, these are a rephrasing of equations (2.5) and (2.6) in \cite{DIZ}.
\end{proof}
\begin{lemma}
\label{lemmaYdifferentiable}
The function $\widehat Y$ defined by \eqref{def:Yhat} is differentiable with respect to $x$ and $t$ for any $z \in \mathbb C\backslash \mathbb R$, and so are its boundary values $\widehat Y_\pm\left(z\right)$, for any $z \in \mathbb R\backslash{\left\{r_1,\ldots,r_k\right\}}$. Moreover $\partial_x \widehat Y$ and $\partial_t \widehat Y$ solve the RH problem detailed below, where $\partial$ denotes the derivative either with respect to $x$ or $t$. 
\end{lemma}

\subsubsection*{RH problem for $\partial \widehat Y$}
\begin{enumerate}[label=(\alph*)]
\item $\partial \widehat Y$ is analytic in $\mathbb C\setminus \mathbb R$,
\item $\partial \widehat Y$ has boundary values $\partial \widehat Y_\pm$ on the real line in $L^2\left(\mathbb R\right)$ and which are continuous  except possibly at the points $r_1,\ldots, r_k$, and they are related by
\begin{equation}\label{eq:jumppartialY}
	\partial \widehat Y_+=\partial\left(\widehat Y_-{J_{\widehat Y}}\right),
\end{equation}
\item $ \partial \widehat Y(z) = \dfrac{\partial \left(t^{-2/3}Y_1\right)}{z} + \mathcal O(z^{-2})$ as $z \to \infty$, with $Y_1$ defined as in \eqref{asY}.
\end{enumerate}

\begin{proof}
	We start by observing that $J_{\widehat Y}(z)$ is differentiable with respect to $x$ and $t$, because the Airy function is smooth. Moreover, $\partial J_{\widehat Y}$ is in $L^2(\mathbb R)$ for any $x\in\mathbb R$, $t>0$. Indeed, the form of the jump matrix $J_{Y}$ (see \eqref{jumpY1} and the equation thereafter) and \eqref{def:Yhat} imply that $\partial J_{\widehat Y}$ is in $L^2(\mathbb R)$ provided that 
\[\sigma(r)\partial \left({\rm Ai}^2(\zeta(r))\right),\qquad \sigma(r)\partial \left({\rm Ai}'^2(\zeta(r))\right),\qquad \sigma(r)\partial \left({\rm Ai}(\zeta(r)){\rm Ai}'(\zeta(r))\right)\]
are in $L^2(\mathbb R)$.
The latter follows from the asymptotics of the Airy function and its derivative together with the Airy differential equation ${\rm Ai}''(\zeta)=\zeta {\rm Ai}(\zeta)$, \eqref{zetaparametrization}, and the third part of Assumptions \ref{assumptions}.

Returning to \eqref{eq:Cauchy3}, we obtain that $\widehat Y_-$ is differentiable with respect to $x$ and $t$, and satisfies the equation
	\begin{equation*}
		\partial \widehat Y_- = C_{J_Y}\left(\partial \widehat Y_-\right) + C_-\left(\widehat Y_-\partial {J_{\widehat Y}}\right) \quad \text{or, equivalently, } \quad \partial \widehat Y_- = \left(1 - C_{J_{\widehat Y}}\right)^{-1} C_-\left(\widehat Y_-\partial {J_{\widehat Y}}\right).
	\end{equation*}
	Consequently, also the right hand side of equation \eqref{eq:Cauchy1} is differentiable with respect to $x$ and $t$, and so is $\widehat Y$. In particular, 
	\begin{equation}\label{eq:Cauchy5}
		\partial \widehat Y = C\left(\partial \widehat Y_-\left({J_{\widehat Y}} - I\right) + \widehat Y_-\partial {J_{\widehat Y}}\right).
	\end{equation}
	Taking into account that $C_+ - C_- = I$ and that $\partial\widehat Y_\pm = C_\pm\left(\partial \widehat Y_-\left({J_{\widehat Y}} - I\right) + \widehat Y_-\partial {J_{\widehat Y}}\right)$, we recover the jump condition \eqref{eq:jumppartialY}. As for the asymptotics at infinity of $\partial \widehat Y$, rewriting \eqref{eq:Cauchy5} as
\begin{align*}
		\partial \widehat Y (z) &= -\frac{1}{2 \pi \i } \int_\R \frac{\partial\left(\widehat Y_-(J_{\widehat Y} - I) \right)(\xi)}{z - \xi} \d\xi
		\\
		&=		-\frac{1}{2 \pi \i z} \int_\R \partial\left(\widehat Y_-(J_{\widehat Y} - I) \right)(\xi) \d\xi-\frac{1}{2 \pi \i z^2} \int_\R \partial\left(\widehat Y_-(J_{\widehat Y} - I) \right)(\xi) \frac{\xi\d\xi}{1-\frac{\xi}{z}}
		\\
		&=-\frac{1}{2 \pi \i z} \int_\R \partial\left(\widehat Y_-(J_{\widehat Y} - I) \right)(\xi) \d\xi+\O\left(z^{-2}\right)
\end{align*}
as $z\to\infty$.
We can conclude by  observing that the term in $z^{-1}$ is indeed the derivative of $t^{-2/3}Y_1$, as it is immediately seen from \eqref{eq:Cauchy1} and condition (c) in the RH problem for $\widehat Y$.
\end{proof}

We can now prove the differentiability of $\Psi$, as defined in \eqref{eq:defPsi}, with respect to $x$ and $t$. 

\begin{lemma}
For any $z\in\mathbb C\setminus\mathbb R$,	$\Psi(z)$ is differentiable with respect to $x$ and $t$, and its derivative satisfies the RH problem detailed below, where $\partial$ denotes the derivative either with respect to $x$ or $t$.
\end{lemma}
\subsubsection*{RH problem for $\partial\Psi$}
\begin{enumerate}[label=(\alph*)]
\item $\partial \Psi$ is analytic in $\mathbb C\setminus \mathbb R$,
\item $\partial \Psi$ has boundary values $\partial\Psi_\pm$ which are $L^2$ on any real compact, and continuous on the real line except possibly at $r_1,\ldots, r_k$, and they are related by
\begin{equation}\label{jumppartialPsi}\partial\Psi_+\left(z\right)=\left(\partial\Psi_-\left(z\right)\right)\begin{pmatrix}1&1-\s(z)\\0&1\end{pmatrix},
\end{equation}
\item As $z\to\infty$,
\begin{align}
\pa	\Psi\left(z\right) = \Bigg(\left(I + \frac{1}{z} \begin{pmatrix} q & -\i r \\ \i p & -  q \end{pmatrix} + \mathcal O\left(\frac{1}{z^2}\right) \right) z^{\frac{1}{4} \sigma_3} A^{-1} \partial \left(-\frac{2}{3}tz^{3/2} + x z^{1/2}\right)\s_3 {\rm e}^{\left(-\frac{2}{3}tz^{3/2} + x z^{1/2}\right)\sigma_3} \nonumber \\
	+ \left( \frac{1}{z} \partial \begin{pmatrix} q & -\i r \\ \i p & -  q \end{pmatrix} + \mathcal O\left(\frac{1}{z^2} \right)\right) z^{\frac{1}{4} \sigma_3} A^{-1}  {\rm e}^{\left(-\frac{2}{3}tz^{3/2} + x z^{1/2}\right)\sigma_3}  \Bigg) \times\begin{cases}I,&|\arg z|<\pi-\delta,\\
\begin{pmatrix}1&0\\\mp 1&1\end{pmatrix},&\pi-\delta<\pm \arg z<\pi,\end{cases}\label{eq:asymppartialPsiI}
\end{align}
for any $0<\delta<\pi/2$; here the principal branches of $z^{\frac{1}4\sigma_3}$ and $z^{1/2}$ are taken, analytic in $\mathbb C \backslash \left(-\infty,0\right]$ and positive for $z>0$.
\end{enumerate}
\begin{proof}
By definition, see equation \eqref{eq:defPsi} and \eqref{def:Yhat},
	$$\Psi(z) = \begin{pmatrix}1& \frac{\i x^2}{4t}\\0&1
	\end{pmatrix}t^{-\frac{1}{6} \sigma_3}\widehat Y\left(z\right)\Phi_{\Airy}\left(\zeta\left(z\right)\right)$$
and every factor on the right hand side of \eqref{eq:defPsi} is differentiable with respect to $x$ and $t$. 
Indeed, for $\widehat Y(z)$ this follows from Lemma \ref{lemmaYdifferentiable}, and for $\Phi_\Ai(\zeta(z))$ we have an explicit expression in terms of the Airy function and its derivative. Hence, also $\Psi$ is differentiable with respect to $x$ and $t$. The asymptotic expansion \eqref{eq:asymppartialPsiI}, which can be formally obtained by differentiating \eqref{eq:Psiasymp}, is rigorously proven by differentiating the definition of $\Psi$ (equation \eqref{eq:defPsi}) and by using the asymptotic expansion of $\widehat Y(z)$, $\Phi_\Ai(\zeta(z))$ and their derivatives at infinity, together with the explicit formulas \eqref{eq:pqalphabeta}. As for the jump condition, this is the same as the jump condition for $\Psi\left(z\right)$ as $\Psi_-^{-1}\left(z\right)\Psi_+\left(z\right)$ does not depend on $x$ or $t$ (see equation \eqref{proofjump}).
\end{proof}

Define
\begin{equation}\label{eq:defTheta}
	\Theta\left(z;x,t\right) := {\rm e}^{\frac{\pi \i}4 \sigma_3}\begin{pmatrix} 1 & -\i p\\ 0 & 1 \end{pmatrix}\Psi\left(z;x,t\right){\rm e}^{-\frac{\pi \i}4 \sigma_3},
\end{equation}
where $p = p_\s\left(x,t\right)$ is given by \eqref{eq:Psiasymp}.

\begin{proposition}\label{prop:KdVLaxpair}
$\Theta = \Theta(z;x,t)$ defined by \eqref{eq:defTheta} satisfies the Lax pair equations
\begin{align}
\nonumber
	\partial_x \Theta &= B\Theta, \qquad B=B\left(z;x,t\right) := \begin{pmatrix} 0 & -z + 2u \\ -1 & 0 \end{pmatrix}, 
\\ \label{eq:Laxpair}
	\partial_t \Theta &= C\Theta, \qquad C =C\left(z;x,t\right):= \renewcommand*{\arraystretch}{1.4}\begin{pmatrix} -\frac{1}3 \partial_x u & \frac{2}3 z^2 - \frac{2}3 z u - \frac{4}3 u^2 - \frac{1}3 \partial_x^2 u \\ \frac{2}3 z + \frac{2}3 u & \frac{1}3 \partial_x u\end{pmatrix},
\end{align}
with $p, q$ as in \eqref{eq:Psiasymp}, and
\begin{equation}\label{upQ}
u = u(x,t) := \partial_x p(x,t)=-2q(x,t)-p^2(x,t).
\end{equation}
\end{proposition}
\begin{proof}
Consider the matrix-valued functions \[\widetilde B := \left(\partial_x\Psi\left(z\right)\right) \Psi^{-1}\left(z\right),\qquad \widetilde C := \left(\partial_t\Psi\left(z\right)\right)\Psi^{-1}\left(z\right).\] By \eqref{jumpPsi} and \eqref{jumppartialPsi}, $\widetilde B$ and $\widetilde C$ have no jumps on $\mathbb R\setminus\left\lbrace r_1,\ldots, r_k\right\rbrace$. Hence, $\widetilde B$ and $\widetilde C$ are analytic functions with possibly isolated singularities at $r_1,\ldots,r_k$. The fact that the boundary values of $\Psi$ and $\partial\Psi$ are locally $L^2$ implies that the singularities are removable, and that $\widetilde B$ and $\widetilde C$ can be extended to entire functions.
From \eqref{eq:Psiasymp} and \eqref{eq:asymppartialPsiI}, we infer that $\widetilde B$ and $\widetilde C$ are polynomials in $z$
of degree $1$ and $2$ respectively, which take the following form:
\begin{equation*}
	\widetilde B := \begin{pmatrix} p & \i z + 2 \i q \\ \\ -\i & -p\end{pmatrix}, \quad \quad \widetilde 
	C := \begin{pmatrix} -\frac{2}3 z p + \frac{2}3 \partial_x q & -\frac{2}3 \i z^2 - \frac{4}3 \i z q - \frac{2}3 \i \partial_x r \\ \\ \frac{2}3 \i z - \frac 43 \i q - \frac{2}3 \i p^2 & \frac{2}3 z p - \frac{2}3 \pa_xq
	\end{pmatrix}.
\end{equation*}
Moreover, while using \eqref{eq:Psiasymp} and \eqref{eq:asymppartialPsiI} to compute $\widetilde B$, we know that the $1/z$ term in the large $z$ expansion must vanish, and this gives us the last identity in \eqref{upQ}.
We can then compute
\begin{equation}\label{eq:conjugBC}
	B = R \widetilde B R^{-1} + \left(\partial_x R\right) R^{-1}, \quad C = R \widetilde C R^{-1} + \left(\partial_t R\right) R^{-1}, \quad \quad R :=  {\rm e}^{\frac{\pi\i}4 \sigma_3}\begin{pmatrix} 1 & -\i p\\ 0 & 1 \end{pmatrix},
\end{equation}
{and this yields equations \eqref{eq:Laxpair} by \eqref{upQ}.}
\end{proof}

As a first consequence of the differential equation in $x$, if we denote $\Theta$ as
\begin{equation}\label{def:phiTheta}
	\Theta\left(z;x,t\right) =\sqrt{2\pi}
	\renewcommand*{\arraystretch}{1.2}	
	\begin{pmatrix}
		\partial_x\phi_\s\left(z;x,t\right) & -\partial_x\wt\phi_\s\left(z;x,t\right)\\
		-\phi_\s\left(z;x,t\right) & \wt\phi_\s\left(z;x,t\right)
	\end{pmatrix},
\end{equation}
then both $\phi_\s,\wt\phi_\s$ are solutions of the Schr\"odinger equation \eqref{eq:Schrodinger} with $u$ in place of $u_\s$.
By \eqref{eq:Psiasymp} and \eqref{eq:defTheta}, as $z \to  \infty$ with $|\arg z|<\pi-\delta$, $\phi_\s$ has the following asymptotics:
\begin{equation*}
	\phi_\s\left(z\right) = \frac{1}{2\sqrt{\pi} z^{\frac{1}4}} \left(1 + \mathcal O\left(z^{-1/2}\right) \right){\rm e}^{-\frac{2}3 tz^{3/2} + xz^{1/2}}\sim t^{1/6}{\rm Ai}\left(t^{2/3}z -xt^{-1/3}\right).
\end{equation*}

As a second consequence, since $\partial_x\partial_t\Theta=\partial_t\partial_x\Theta$, we have the compatibility condition
$$
\partial_t B - \partial_x C + [B,C] = 0.
$$
Substituting the expressions \eqref{eq:Laxpair} in this equation, we obtain after a straightforward computation that $u$ solves the KdV equation \eqref{eq:KdV}. 

In conclusion, we have proven that $u$, defined as in \eqref{upQ}, satisfies the KdV equation \eqref{eq:KdV} and found a function $\phi_\s$ solving \eqref{eq:Schrodinger} and \eqref{eq:eveq}, with potential $u$ in place of $u_\s$, with asymptotics \eqref{eqthm:phiAi}. To complete the proof of Theorem \ref{theorem:main}, it is then enough to show that actually $u=u_\s$ defined in \eqref{def:uQ}, which we will do in the next section.

\section{Derivation of the integro-differential Painlev\'e II equation}\label{sec:4}

We first restate Proposition \ref{prop:PsiQ} by expressing the $x$-logarithmic derivative of the Fredholm determinant $Q_\sigma\left(x,t\right)$, defined in \eqref{def:Q}, in terms of $\Theta$ rather than $\Psi$. We have
\begin{equation*}
\partial_x \log Q_\sigma\left(x,t\right) = -\frac{1}{2 \pi t} \int_{\mathbb R} \left(\Theta^{-1}\left(z\right)\left(\pa_z\Theta\left(z\right)\right) \right)_{21} \d\s(z) = \frac{1}{t} \int_{\mathbb R} \left( \phi_x\left(z\right)\phi_z\left(z\right) -\phi\left(z\right)\phi_{xz}\left(z\right) \right) \d \s(z),
\end{equation*}
where the first equality follows from \eqref{eq:PsiQ} along with the definition \eqref{eq:defTheta} of $\Theta$, and the second one from \eqref{def:phiTheta}. Here we denote for brevity $\phi=\phi_\s$ and derivatives by subscripts.

Taking one more $x$-derivative and using the Schr\"odinger equation $\phi_{xx}=(z-2u)\phi$, we obtain
\be
\label{eq:proofQphi}
\partial^2_x \log Q_\sigma\left(x,t\right) = \frac{1}{t} \int_{\mathbb R} \left( \phi_{xx}\left(z\right)\phi_z\left(z\right) -\phi\left(z\right)\phi_{xxz}\left(z\right) \right) \d \s(z)=-\frac 1t \int_{\mathbb R} \phi^2\left(z\right) \d \s(z).
\ee

To complete the identification of $u$ with $u_\s$ we prove the following.

\begin{proposition}
\label{prop:proofu2}
Let $\phi_\s\left(z;x,t\right)$ be defined by \eqref{def:phiTheta}, and let $u\left(x,t\right)$ be defined as in \eqref{upQ}. Then, for every $x \in \R$ and $t > 0$, we have the identity
	\begin{equation}\label{equphi}
		\int_\R \phi^2\left(r;x,t\right)\d\s(r) =  \frac x2 - tu(x,t).
	\end{equation}
\end{proposition}

\proof 
It follows from the jump condition of $\Psi$ along the real axis that for all $z\in\R\setminus\{r_1,...,r_k\}$ we have
\begin{equation*}
\Psi_+\left(z\right)\begin{pmatrix} 0 & -\s'\left(z\right) \\ 0 & 0 \end{pmatrix}\Psi_+^{-1}\left(z\right)=\Delta\left[\left(\partial_z\Psi\right)\Psi^{-1}\right]\left(z\right),
\end{equation*}
where we denoted
\[
\Delta\left[\left(\partial_z\Psi\right)\Psi^{-1}\right]\left(z\right):=\left(\left(\partial_z\Psi\right)\Psi^{-1}\right)_+\left(z\right)-\left(\left(\partial_z\Psi\right)\Psi^{-1}\right)_-\left(z\right).
\]
Integrating the (2,1)-entry of this identity, using \eqref{eq:defTheta} and \eqref{def:phiTheta}, we get
\begin{equation*}
-2\pi\int_{\R\setminus\{r_1,\ldots, r_k\}} \sigma'\left(z\right) \phi^2\left(z;x,t\right)\,\d z=\int_{\R\setminus\{r_1,\ldots, r_k\}}\Delta\left[\left(\partial_z\Psi\right)\Psi^{-1}\right]_{2,1}\left(z\right)\,\d z.
\end{equation*}
The integral at the right hand side of this identity can be performed by a contour deformation and this results in
\[
\int_{\R\setminus\{r_1,\ldots, r_k\}}\Delta\left[\left(\partial_z\Psi\right)\Psi^{-1}\right]_{2,1}\left(z\right)\,\d z=2\pi\i\left(\Res_{z=\infty}\left(\left(\partial_z\Psi\right)\Psi^{-1}\right)_{21}+\sum_{j=1}^k\Res_{z=r_j}\left(\left(\partial_z\Psi\right)\Psi^{-1}\right)_{21}\right).
\]
Here we should note that $\left(\partial_z\Psi\right)\Psi^{-1}$ does not necessarily have isolated singularities at the points $r_j$ and at $\infty$, so the above residues should be understood as formal residues, namely the coefficients of the $1/\left(z-r_j\right)$ or $-1/z$ terms in the asymptotic series of $\left(\partial_z\Psi\right)\Psi^{-1}$ around the singularity.
Hence we have established the identity
\be
\label{aaaa}
\i\int_{\R\setminus\{r_1,\ldots,r_k\}} \sigma'\left(z\right) \phi^2\left(z;x,t\right)\,\d z=\Res_{z=\infty}\left(\left(\partial_z\Psi\right)\Psi^{-1}\right)_{21}+\sum_{j=1}^k\Res_{z=r_j}\left(\left(\partial_z\Psi\right)\Psi^{-1}\right)_{21}.
\ee
The residue at $z=\infty$ is computed directly from the expansion \eqref{eq:Psiasymp} as
\be
\label{bbbb}
\Res_{z=\infty}\left(\left(\partial_z\Psi\right)\Psi^{-1}\right)_{21}=\i\left(\frac x2+tp^2\left(x,t\right)+2tq\left(x,t\right)\right)=\i\left(\frac x2-tu\left(x,t\right)\right),
\ee
using \eqref{upQ} in the last step. We can compute the remaining residues by looking at the jump relations for $\Psi$ near $r_j$: $\Psi_+=\Psi_-\begin{pmatrix}1&1-\sigma\\0&1\end{pmatrix}$. Since $\sigma$ has a discontinuity at $r_j$, the boundary values $\Psi_\pm$ cannot be continuous at $r_j$.
Let us define 
\be
\label{eq:defE}
E\left(z\right) := \Psi\left(z\right)\begin{pmatrix}1&-\frac{m_j}{2\pi\i}\log \left(z-r_j\right)\\0&1\end{pmatrix},\qquad m_j=\lim_{\epsilon\to 0_+}\left(\sigma\left(r_j+\epsilon\right)-\sigma\left(r_j-\epsilon\right)\right),
\ee
where we take the principal branch of the logarithm.
We then verify that 
\[E_+\left(z\right)=E_-\left(z\right)\begin{pmatrix}1&1-\widehat\s(z)\\0&1\end{pmatrix},\qquad \widehat\s(z)=\begin{cases}\s(z) &\mbox{for $z>r_j$,}\\
\s(z)+m_j &\mbox{for $z<r_j$.}
\end{cases}\]
In other words, the jump matrix for $E$ is continuous at $r_j$, and the boundary values $E_\pm$ are continuous at $r_j$. Hence, by the definition of $E$, \eqref{eq:defTheta}, and \eqref{def:phiTheta}, we obtain
\be
\label{cccc}
\Res_{z=r_j}\left(\left(\pa_z\Psi\right) \Psi^{-1}\right)_{21}=-\frac{m_j}{2\pi\i}E_{21}^2\left(r_j\right)=-\frac{m_j}{2\pi\i}\Psi_{21}^2\left(r_j\right)=-\i m_j\phi_\sigma^2\left(r_j;x,t\right).
\ee
Substituting \eqref{bbbb} and \eqref{cccc} in \eqref{aaaa}, we obtain \eqref{equphi}.
\qed \\

Comparing \eqref{equphi} with \eqref{eq:proofQphi} we conclude that $u=\pa_x^2\log Q_\s+\frac x{2t}=:u_\s$. In view of this identification, the proof of Theorem \ref{theorem:main} is complete, because the KdV equation and the properties of $\phi_\s$ have already been derived in the previous section, and \eqref{eq:u2} follows from \eqref{equphi}.

\section{Small time asymptotics when $x\leq -Mt^{1/3}$}\label{sec:5}

In this section, we assume that $\sigma$ satisfies Assumptions \ref{assumptions}, and that there are constants $c_1, c_2>0$ such that \eqref{eq:sigmaexpdecay} holds for some $\gamma\in\left(0,1\right]$.

\subsection{Estimates for the jump matrices of $Y$}
We first show that the jump matrices for the RH solution $Y$ are asymptotically close to identity. 
\begin{lemma}\label{lemma:small}
For any $t_0>0$, there exist $M, c>0$ such that we have the estimates
\begin{equation*}
\left\|J_Y-I\right\|_1=\mathcal O\left(\e^{-c\frac{|x|}{t^{1/3}}}\right)
,\quad\left\|J_Y-I\right\|_2=\mathcal O\left(\e^{-c\frac{|x|}{t^{1/3}}}\right)
,\quad
\left\|J_Y-I\right\|_\infty=\mathcal O\left(\e^{-c\frac{|x|}{t^{1/3}}}\right)
,
\end{equation*}
uniformly for $0<t\leq t_0$ and $x<-Mt^{1/3}$, where $\|.\|_1$, $\|.\|_2$, and $\|.\|_\infty$ denote the maximum of the entrywise $L^1\left(\mathbb R\right)$-, $L^2\left(\mathbb R\right)$-, and $L^\infty\left(\mathbb R\right)$-norms.
\end{lemma}

\begin{proof}
Let us consider one entry of the jump matrix $J_Y$ for simplicity, the others are similar. For the (1,2)-entry, we need to prove that
\begin{equation*}
\left\|\s\left(z\left(.\right)\right)\left(\Ai'\left(.\right)\right)^2\right\| =\mathcal O\left(\e^{-c\frac{|x|}{t^{1/3}}}\right),
\end{equation*}
where $\|.\|$ denotes the $L^1$-,$L^2$-,  or $L^\infty$-norm, uniformly in the relevant parameters $x,t$.

\begin{enumerate}
\item For $\zeta\geq \frac{|x|}{2t^{1/3}}$, we have by the asymptotics of the Airy function (note that the power in the exponent is not sharp) the uniform (in $\zeta, x, t$) estimate
\[\Ai'\left(\zeta\right)=\mathcal O\left(\e^{-2\zeta}\right)=\mathcal O\left(\e^{-\zeta}\e^{-\frac{|x|}{2t^{1/3}}}\right),\]
and $|\sigma\left(z\left(\zeta\right)\right)|\leq 1$, and this yields the rough estimate
\begin{equation}\label{eq:estregion1}
\left|\left(J_Y\left(\zeta\right)-I\right)_{21}\right|=\mathcal O\left(\e^{-2\zeta}\e^{-\frac{|x|}{t^{1/3}}}\right).
\end{equation}
\item For $\zeta\leq \frac{|x|}{2t^{1/3}}$, we have
\[\Ai'\left(\zeta\right)=\begin{cases}
\mathcal O\left(1+|\zeta|^{1/4}\right),&\zeta<0,\\
\mathcal O\left(1\right),&\zeta\geq 0,\end{cases}\]
and 
\[z\left(\zeta\right)=t^{-2/3}\left(\zeta+\frac{x}{t^{1/3}}\right)\leq -\frac{|x|}{2t},\]
hence by the exponential decay of $\sigma$,
\[\sigma\left(z\left(\zeta\right)\right)=\mathcal O\left(\e^{-c_2|z\left(\zeta\right)|}\right)=\mathcal O\left(\e^{-\frac{c_2}{2}|z\left(\zeta\right)|}\e^{-\frac{c_2}{2}|z\left(\zeta\right)|}\right)=\mathcal O\left(\min\{\e^{2\zeta},1\} \e^{-\frac{c_2|x|}{4t}}\right).\]
It follows that
\begin{equation}\label{eq:estregion2}
\left|\left(J_Y\left(\zeta\right)-I\right)_{21}\right|=\begin{cases}
\mathcal O\left(\left(1+|\zeta^{1/2}|\right)\e^{2\zeta}\e^{-\frac{c_2|x|}{4t}}\right)=\mathcal O\left(\e^{\zeta}\e^{-\frac{c_2|x|}{4t}}\right)
,&\zeta<0,\\
\mathcal O\left(\e^{-\frac{c_2|x|}{4t}}\right)
,&\zeta\geq 0.
\end{cases}
\end{equation}
\end{enumerate}
Provided that $t\leq t_0$, by \eqref{eq:estregion1} and \eqref{eq:estregion2}, we obtain a constant $c>0$ such that
\[\left\|\left(J_Y-I\right)_{21}\right\|=\mathcal O\left(\|h\| \e^{-\frac{2c|x|}{t^{1/3}}}\right),\qquad h\left(\zeta\right)=\begin{cases}\e^{-|\zeta|},&\zeta<0 \mbox{ and }\zeta>\frac{|x|}{2t^{1/3}},\\
1,&0\leq \zeta\leq \frac{|x|}{2t^{1/3}},\end{cases}\]
where $\|.\|$ denotes the $L^1$-,$L^2$-, or $L^\infty$-norm.
If $M>0$ is sufficiently large, noting that $\|h\|=\mathcal O\left(|x|t^{-1/3}\right)$, we obtain the result.
\end{proof}

\subsection{Asymptotics for $Y$, $u_\sigma$, and $Q_\sigma$}\label{section:asY}

By standard RH techniques, {it follows from Lemma \ref{lemma:small}, i.e.\ from the fact that the jump matrices for $Y$ are asymptotically close to the identity matrix,} that $Y$ satisfies a small-norm RH problem.
{Indeed, Lemma \ref{lemma:small} implies that the norm of the operator $C_{J_Y}$  (acting on $L^2\left(\R,\mathbb R^{2\times 2},\d\zeta\right)$) given by \eqref{def:Cauchyop} is small for $xt^{-1/3}\leq -M$ with $M>0$ large enough, since
\[\|C_{J_Y}\|_2\leq \|C_-\|_2\|J_Y-I\|_\infty.\] Hence, by \eqref{eq:Cauchy3}, we have
\[\|Y_--I\|_2\leq 2\|J_Y-I\|_2\]
for $M$ large enough, and the right hand side decays exponentially fast as $xt^{-1/3}\to -\infty$.
Writing
\begin{equation*}
Y\left(\zeta;x,t\right)=I+\frac{Y^{\left(1\right)}\left(x,t\right)}\zeta+\O\left(\zeta^{-2}\right)
\end{equation*}
as $\zeta\to\infty$, we also have, by \eqref{eq:Cauchy1}, that
\[Y^{\left(1\right)}\left(x,t\right)=\int_{\mathbb R}(Y_-(\zeta)-I)(J_Y(\zeta)-I)\d\zeta +\int_{\mathbb R}(J_Y(\zeta)-I)\d\zeta.
\]
To bound the first term on the right, we can use a Cauchy-Schwarz estimate together with the above $L^2$-bounds for $Y_--I$ and $J_Y-I$; to bound the second term at the right, we can use the $L^1$-bound for $J_Y-I$.
In conclusion, we have that
\[Y^{(1)}(x,t)
=\mathcal O\left(\e^{-c\frac{|x|}{t^{1/3}}}\right)\]
as $xt^{-1/3}\to -\infty$.}
By \eqref{eq:pqalphabeta} and \eqref{asY}, we obtain
\be
\label{eq:aspneg}p_\sigma\left(x,t\right)=\frac{x^2}{4t}+\mathcal O\left(t^{-1/3}\e^{-c\frac{|x|}{t^{1/3}}}\right),
\qquad
q_\sigma\left(x,t\right)=-\frac{x^4+8tx}{32t^{2}}+\mathcal O\left(t^{-2/3}\e^{-c\frac{|x|}{t^{1/3}}}\right).
\ee
Taking $0<t\leq t_0$ fixed and letting $x\to -\infty$, we can use these estimates to integrate \eqref{def:uQ}: {by \eqref{upQ}, we have
\[\int_{-\infty}^x \partial_{x'}^2 \log Q_\sigma(x',t)dx'=\int_{-\infty}^{x}\left(\partial_{x'} p_\sigma(x',t)-\frac{x'}{2t}\right)dx'=\left(p_\sigma(x',t)-\frac{(x')^2}{4t}\right)_{x'=-\infty}^{x}.\]
Using \eqref{eq:aspneg} and the fact that $\partial_x\log Q_\sigma(-\infty,t)=0$ (to see this, recall \eqref{eq:derlog} and note that both $\mathbb K_{\sigma,x,t}$ and $\partial_x\mathbb K_{\sigma,x,t}$ have small $L^2$-norm for $x$ large negative by \eqref{eq:sigmaexpdecay}--\eqref{eq:sigmaprimedecay}) we obtain}
the identity
\begin{equation*}
\partial_x\log Q_\sigma\left(x,t\right)=p_\sigma\left(x,t\right)-\frac{x^2}{4t}.\label{Qp}
\end{equation*}
Substituting the asymptotics \eqref{eq:aspneg} at the right hand side, and integrating in $x=x'$ from $-\infty$ to $x$, we get
\begin{equation*}
\log Q_\sigma\left(x,t\right)=\mathcal O\left(\e^{-c|x|t^{-1/3}}\right),\label{Qp2}
\end{equation*}
which is part {\it (i)} of Theorem \ref{thmQ}.
By \eqref{upQ}, we obtain uniformly for $x,t$ in the relevant regions that
\[u\left(x,t\right)=\frac{x}{2t}+\mathcal O\left(t^{-2/3}\e^{-c\frac{|x|}{t^{1/3}}}\right).\]
After renaming $c\mapsto 2c$, this implies part {\it (i)} of Theorem \ref{thminitialvalue}.

\section{Small time asymptotics when $-Mt^{1/3}\leq x\leq Mt^{1/3}$}\label{sec:proofii}

Like in the previous section, we assume that $\sigma$ satisfies Assumptions \ref{assumptions} with $\gamma\in\left(0,1\right]$ and that in addition condition \eqref{eq:sigmaexpdecay} and \eqref{eq:sigmaprimedecay} hold for some $c_1,c_2,c_3,C>0$. Furthermore, we assume that $t>0$ is sufficiently small and that $-Mt^{1/3}\leq x\leq Mt^{1/3}$, where $M>0$ is arbitrarily large, and we recall the RH problem for $\Psi=\Psi_\s$, which depends on the function $\sigma$.

\subsection{Change of variable $z\mapsto w$}\label{par:changeofvariable}
For our purposes it is convenient to introduce a scaled variable
\be\label{def:w}
w=t^{2/3}z=\zeta\left(z\right)+xt^{-1/3}
\ee
and to define
\be
\label{eq:defwtPsi}
\wt\Psi_\s\left(w\right)=\wt\Psi_\s\left(w;x,t\right):=t^{\frac{\s_3}6}\Psi_\s\left(z=wt^{-2/3};x,t\right).
\ee
It follows from the RH characterization of $\Psi_\s$ that $\wt\Psi_\s$ solves the following RH problem.

\subsubsection*{RH problem for $\wt\Psi_\s$}
\begin{enumerate}[label=(\alph*)]
\item $\wt\Psi_\s$ is analytic in $\mathbb C\setminus \mathbb R$.
\item $\wt\Psi_\s$ has boundary values $\wt\Psi_{\s,\pm}$ which are $L^2$ on any compact real set and which are continuous on the real line except at the points $t^{2/3}r_1, \ldots ,t^{2/3}r_k$, and they are related by
\begin{equation}\label{jumpwtPsi}
\wt\Psi_{\s,+}\left(w\right)=\wt\Psi_{\s,-}\left(w\right)\begin{pmatrix}1&1-\sigma\left(wt^{-2/3}\right)\\0&1\end{pmatrix}.
\end{equation}
\item As $w\to\infty$, there exist functions $p= p_{\sigma}\left(x,t\right), q= q_{\sigma}\left(x,t\right)$ and $r= r_{\sigma}\left(x,t\right)$ such that $\wt\Psi_\s$ has the asymptotic behavior
\begin{multline}
\label{eq:wtPsiasymp}
\wt\Psi_\s\left(w\right) = \left( I + \frac{t^{2/3}}{w} \begin{pmatrix} q & -\i t^{1/3} r \\ \i t^{-1/3} p & -  q \end{pmatrix} + \mathcal O\left(\frac{1}{w^2}\right) \right)\\
\times\ w^{\frac{1}{4} \sigma_3} A^{-1}{\rm e}^{\left(-\frac{2}{3}w^{3/2}+xt^{-1/3}w^{1/2}\right)\sigma_3}\times\begin{cases}I,&|\arg z|<\pi-\delta,\\
\begin{pmatrix}1&0\\\mp 1&1\end{pmatrix},&\pi-\delta<\pm \arg z<\pi,\end{cases}
\end{multline}
for any $0<\delta<\pi/2$; here the principal branches of $w^{\frac{1}4\sigma_3}$, $w^{3/2}$, and $w^{1/2}$ are taken, analytic in $\mathbb C \backslash \left(-\infty,0\right]$ and positive for $w>0$.
\end{enumerate}

Note that the functions $p,q,r$ in (c) are the same as the ones that appear in the corresponding property of the RH characterization of $\Psi_\s$.

As $t\to 0$ we have the pointwise limit
\begin{equation*}
\lim_{t\to 0}\sigma\left(wt^{-2/3}\right)= \s_\g\left(w\right)
\end{equation*}
for all $w\not = 0$, where we denote $\s_\g=\g \chi_{\left[0,\infty\right)}$ as in Examples \ref{examples}.
This suggests that $\wt\Psi_\step\left(w\right)$ might be a good approximation to $\wt\Psi_\s$ away from $0$. Motivated by this observation, we will now take a closer look at the RH problem in the case where $\s=\s_\g$.

\subsection{RH problem for $\wt\Psi_{\step}$ and connection with the Painlev\'e II equation}

We have the self-similarity relation
\begin{equation*}
\wt\Psi_\step\left(w;x,t\right)=\wt\Psi_\step\left(w;xt^{-1/3},1\right),
\end{equation*}
as it is straightforward to check that both sides satisfy the same RH conditions. In particular, by comparing the asymptotic behaviors as $w\to\infty$, we obtain the self-similarity properties
\be
\label{eq:pqrstep}
q_\step(x,t)=t^{-2/3}q_\step\left(xt^{-1/3},1\right),\ p_\step(x,t)=t^{-1/3}p_\step\left(xt^{-1/3},1\right),\ r_\step(x,t)=t^{-1}r_\step\left(xt^{-1/3},1\right).
\ee
Moreover, the RH problem for $\wt\Psi_\step$ is equivalent to a RH problems associated
to the Painlev\'e XXXIV equation, which is related to the Painlev\'e II equation. More precisely, the solution $\wt\Psi_\step$ can be constructed in terms of the solution $\Psi_0$ to the RH problem from \cite[Section 3.3]{BCI}: 
indeed we have
\begin{equation}
\label{eq:astildePsi}
\wt\Psi_\step\left(w;x,t\right)=
\begin{cases}
\Psi_0\left(w;\tau=-xt^{-1/3}\right)&\mbox{for $|\arg w|<2\pi/3$,}\\
\Psi_0\left(w;\tau=-xt^{-1/3}\right)\left(\begin{matrix}1&0\\-1&1\end{matrix}\right)&\mbox{for $2\pi/3<\arg w<\pi$,}\\
\Psi_0\left(w;\tau=-xt^{-1/3}\right)\left(\begin{matrix}1&0\\1&1\end{matrix}\right)&\mbox{for $-\pi<\arg w<-2\pi/3$.}
\end{cases}
\end{equation}

We now rely on some results about $\Psi_0$ proved in \cite{BCI}, which imply identities for the functions $p_{\step}, q_{\step}, r_{\step}$. From \cite[equations (3.26-27), (3.38-39), and (3.41)]{BCI}, we have
\begin{equation*}
\pa_\xi p_{\step}(\xi,1)=-2q_{\s_\g}(\xi,1)-p^2_{\s_\g}(\xi,1)=\frac \xi 2 -y_\g^2(-\xi),
\end{equation*}
where $y_\g^2$ (denoted $y_\g$ in \cite{BCI}) is a solution to the Painlev\'e XXXIV equation, and $y_\g$
is the Painlev\'e II solution characterized by \eqref{eq:PII}. In view of \eqref{upQ} and \eqref{eq:pqrstep}, we also have
\be
\label{eq:stepKdVsolution2}
\pa_x p_\step(x,t)=-2q_\step(x,t)-p_\step^2(x,t)=
\frac x{2t} - \frac 1{t^{2/3}} y_\g^2\left(-xt^{-1/3}\right).
\ee
Using this identity and the asymptotics \eqref{eq:aspneg} to integrate $\partial_xp_\step(x,t)-\frac{x}{2t}$ in $x$ between $-\infty$ and $x$, we obtain
\be\label{eq:idPIIp2}
p_\step(x,t)=\frac{x^2}{4t}-t^{-1/3}\int_{-xt^{-1/3}}^{+\infty}y_\gamma^2(v)\d v.
\ee

\subsection{Asymptotics for $\wt\Psi_\s$}

It will turn out that $\wt\Psi_\step$ is a good approximation to $\wt\Psi_\s$ as $t\to 0$ as long as $w$ is not too small. For $w$ small, we need to construct a local parametrix different from $\wt\Psi_\step$.

We define the local parametrix as follows,
\be\label{eq:parametrix2}
P\left(w\right)=\wt\Psi_\step\left(w\right)\begin{pmatrix}1 & a\left(w\right) \\ 0 & 1\end{pmatrix}\qquad \mbox{for $|w|\leq 1$,}
\ee
where 
\be
\label{eq:defa}
a\left(w\right)=\frac 1{2\pi\i}\left(\int_{-\infty}^0\frac{-\s\left(w't^{-2/3}\right)}{w'-w}\d w'+\int_0^{+\infty}\frac{\g-\s\left(w't^{-2/3}\right)}{w'-w}\d w'\right)\qquad \mbox{for $w\notin\mathbb R$}.
\ee
Then we define
\begin{equation*}
S\left(w\right)=\begin{cases}
\wt\Psi_\s\left(w\right)\wt\Psi_\step^{-1}\left(w\right), & \text{if} \; |w|>1 \; \text{and}\; w\in\C\setminus\R, 
\\
\wt\Psi_\s\left(w\right)P^{-1}\left(w\right), & \text{if} \; |w|<1\; \text{and}\; w\in\C\setminus\R.
\end{cases}
\end{equation*}
We will now show that $S$ solves the following RH problem.

\subsubsection*{RH problem for $S$}
\begin{enumerate}[label=(\alph*)]
\item $S$ is analytic in $\mathbb C\setminus \G$, where $\G$ is the oriented contour in the complex plane given by the union of $\R\cap\{|w|>1\}$, oriented from left to right, and the unit circle, oriented counter-clockwise.
\item $S$ has boundary values $S_\pm$ which are continuous on $\G$; according to the chosen orientation of the unit circle, $S_+$ corresponds to the boundary value from the inside, and $S_-$ to that from the outside. Moreover, the boundary values are related by
\be
S_+\left(w\right)=S_-\left(w\right)J_S\left(w\right),
\ee
with
\begin{equation*}
J_S\left(w\right)=
\begin{cases}
\wt\Psi_\step\left(w\right)P^{-1}\left(w\right),&\text{if}\;|w|=1,
\\
\wt\Psi_{\step,-}\left(w\right)\begin{pmatrix} 1 & \g \chi_{\left[0,\infty\right)}\left(w\right)-\s\left(wt^{-2/3}\right) \\ 0 & 1\end{pmatrix} \wt\Psi_{\step,-}^{-1}\left(w\right),&\text{if}\;w<-1\;\text{or}\;w>1.
\end{cases}
\end{equation*}
\item For $|\arg w|<\pi$, we have uniformly $$S\left(w\right) = I+\O\left(w^{-1}\right)\; \text{as} \; w\to\infty.$$
\end{enumerate}

\begin{figure}[htbp]
\centering
\hspace{1cm}
\begin{tikzpicture}
\draw[postaction={decorate,decoration={markings,mark=at position .5 with {\arrow[line width=1pt ]{>}}}}] (-4,0) -- (-1,0) ;
\draw[postaction={decorate,decoration={markings,mark=at position .5 with {\arrow[line width=1pt ]{>}}}}] (1,0) -- (4,0) ;
\draw [postaction={decorate,decoration={markings,mark=at position .2 with {\arrow[line width=1pt ]{>}},mark=at position .7 with {\arrow[line width=1pt ]{>}}}}] (0,0) circle (1cm);
\node at (-2,.2) {$+$};
\node at (-2,-.2) {$-$};
\node at (2,.2) {$+$};
\node at (2,-.2) {$-$};
\node at (0,.8) {$+$};
\node at (0,1.2) {$-$};
\node at (0,-.8) {$+$};
\node at (0,-1.2) {$-$};
\end{tikzpicture}
\caption{The oriented jump contour $\G$.}
\label{figcontour}
\end{figure}
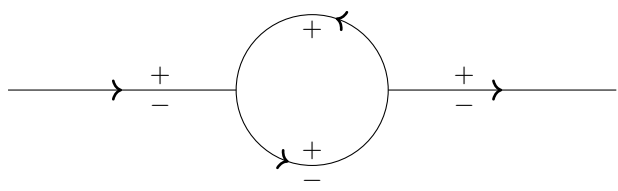

Since $S$ is built out of functions that are {analytic in $w$ away from $\R\cup\{|w|=1\}$, it is clear that $S$ is analytic in this domain. We now show that, in addition,} $S$ extends analytically across $w\in\left(-1,1\right)$. To see this, let us first note the relation
\be
\label{eq:jumpa}
a_+\left(w\right)-a_-\left(w\right)=\g \chi_{\left[0,\infty\right)}\left(w\right)-\s\left(wt^{-2/3}\right)
\ee
for all $w\in\R$; this is a consequence of the Plemelj formula applied to the singular integral \eqref{eq:defa}. Then, for all $w\in\left(-1,1\right)$, we have
\begin{align}
\nonumber
S_+\left(w\right)&=\wt\Psi_{\s,+}\left(w\right)P_{0,+}^{-1}\left(w\right)=\wt\Psi_{\s,+}\left(w\right)\begin{pmatrix}1 & -a_+\left(w\right) \\ 0 & 1\end{pmatrix}\wt\Psi_{\step,+}^{-1}\left(w\right)
\\
\nonumber
&=\wt\Psi_{\s,-}\left(w\right)\begin{pmatrix}1 & 1-\s\left(wt^{-2/3}\right) \\ 0 & 1 \end{pmatrix} \begin{pmatrix}1 & -a_+\left(w\right) \\ 0 & 1\end{pmatrix} \begin{pmatrix}1 & \g \chi_{\left[0,\infty\right)}\left(w\right) \\ 0 & 1 \end{pmatrix} \wt\Psi_{\step,-}^{-1}\left(w\right)
\\
&
=\wt\Psi_{\s,-}\left(w\right)\begin{pmatrix}1 & -a_-\left(w\right) \\ 0 & 1 \end{pmatrix} \wt\Psi_{\step,-}^{-1}\left(w\right)
=\wt\Psi_{\s,-}\left(w\right) P_{0,-}^{-1}\left(w\right)
=S_-\left(w\right),
\end{align}
{which proves that $S$ is meromorphic with possibly isolated singularities at the points $w=t^{2/3}r_j$ and $w=0$. Moreover, near any of the points $w=t^{2/3}r_j$ we have the representation (as in \eqref{eq:defE})
\be
\wt\Psi_\s\left(w\right) = E_j\left(w\right)\begin{pmatrix}1&\frac{m_j}{2\pi\i}\log \left(t^{-2/3}w-r_j\right)\\0&1\end{pmatrix},\qquad m_j=\lim_{\epsilon\to 0_+}\left(\sigma\left(r_j+\epsilon\right)-\sigma\left(r_j-\epsilon\right)\right),
\ee
for a locally holomorphic function $E_j$ in the neighborhood of $w=t^{2/3}r_j$ and therefore
\be
S\left(w\right)=I+\wt\Psi_\s\left(w\right)\begin{pmatrix}0 & -a\left(w\right) \\ 0 & 0\end{pmatrix}\wt\Psi_\step^{-1}\left(w\right)=I+E_j\left(w\right)\begin{pmatrix}0 & -a\left(w\right) \\ 0 & 0\end{pmatrix}\wt\Psi_\step^{-1}\left(w\right)
\ee
which is holomorphic near $w=t^{2/3}r_j$; a similar argument shows that also $w=0$ is a removable singularity of $S$, and this proves (a).
}

To prove (b), note that the continuity of $S_\pm$ on $\G$ follows from the analogous condition for $\wt\Psi_\s$, since the points $t^{2/3}r_1, \ldots ,t^{2/3}r_k$ will be inside the unit disk as soon as $t$ is small enough; the jump relation follows from the computations
\be
S_+\left(w\right)=\wt\Psi_\s\left(w\right) P^{-1}\left(w\right)=\wt\Psi_\s\left(w\right) \wt\Psi_\step^{-1}\left(w\right) \wt\Psi_\step\left(w\right) P^{-1}\left(w\right)=S_-\left(w\right)\wt\Psi_\step\left(w\right) P^{-1}\left(w\right)
\ee
for $|w|=1$, and
\begin{align}
S_+\left(w\right)&=\wt\Psi_{\s,+}\left(w\right)\wt\Psi_{\step,+}^{-1}\left(w\right)=\wt\Psi_{\s,-}\left(w\right)\begin{pmatrix}1 & \g \chi_{\left[0,\infty\right)}\left(w\right)-\s\left(wt^{-2/3}\right) \\ 0 & 1 \end{pmatrix}
\wt\Psi_{\step,-}^{-1}\left(w\right)
\\
&=S_-\left(w\right)\wt\Psi_{\step,-}\left(w\right)
\begin{pmatrix}1 & \g \chi_{\left[0,\infty\right)}\left(w\right)-\s\left(wt^{-2/3}\right) \\ 0 & 1 \end{pmatrix}
\wt\Psi_{\step,-}^{-1}\left(w\right)
\end{align}
for $w<-1$ or $w>1$.
Finally, condition (c) follows from the corresponding condition in the RH problem for $\wt\Psi_{\sigma}$.

We now show that the jump matrix $J_S$ is small in the relevant regime, provided that $\s$ also satisfies the decay conditions \eqref{eq:sigmaexpdecay} and \eqref{eq:sigmaprimedecay}.

\begin{lemma}\label{lemma:RH2jump}
{There exists $\epsilon>0$ such that for any $M>0$ the following estimates hold true uniformly in $0<t<\epsilon, |x|\leq Mt^{1/3}$;}
\begin{equation*}
\|J_S-I\|_1=\O\left(t^{2/3}\right),
\quad
\|J_S-I\|_2=\O\left(t^{2/3}\right),
\quad
\|J_S-I\|_\infty=\O\left(t^{2/3}\right),
\end{equation*}
where $\|.\|_1$, $\|.\|_2$, and $\|.\|_\infty$ denote the maximum of the entrywise $L^1\left(\G\right)$-, $L^2\left(\G\right)$- and $L^\infty\left(\G\right)$-norms, respectively.
\end{lemma}

\begin{proof}
On the unit circle $|w|=1$ we have
\begin{equation*}
J_S\left(w\right)-I=a\left(w\right)\wt\Psi_\step\left(w\right)\begin{pmatrix}0 & 1 \\ 0 & 0 \end{pmatrix}\wt\Psi_\step^{-1}\left(w\right).
\end{equation*}
As $t\to 0$, and uniformly on $|w|=1$, we have $\wt\Psi_\step(w)=\O(1)$, because the variable $w$ and parameter $\tau=-xt^{-1/3}$ of the matrix $\Psi_0$ in \eqref{eq:astildePsi} remain bounded (in particular, $\tau\in(-M,0)$ because of our assumptions).
Hence, on the unit circle $|w|=1$,
\begin{equation*}
J_S\left(w\right)-I=a\left(w\right)\mathcal O\left(1\right).
\end{equation*}
To estimate $a$ we rewrite \eqref{eq:defa} by the change of variable $z'=w't^{-2/3}$ to obtain
\be
\label{eq:estimatea}
|a\left(w\right)|\leq\frac{t^{2/3}}{2\pi}\int_\R\frac{|\step(z')-\s\left(z'\right)|}{|t^{2/3}z'-w|}\d z'
\leq t^{2/3}\frac {c_1}{|\Im w|}\int_\R\e^{-c_2|z'|}\d z'\leq\frac{2c_1}{c_2\delta}t^{2/3}
\ee
assuming $|\Im w|\geq\delta$ and using \eqref{eq:sigmaexpdecay}. This proves that $a(w)=\O\left(t^{2/3}\right)$ uniformly in $|w|=1,|\Im w|\geq\delta>0$. However we can provide a similar estimate when $|w|=1$ and $|\Im w|<\delta$; by symmetry, we can focus only on the connected component containing $w=-1$, and write
\begin{equation*}
a(w)=\frac{t^{2/3}}{2\pi\i}\left(I_1+I_2\right),
\end{equation*}
where
\be
\label{eq:I1I2}
I_1:=\int_{\left(-\infty,-2t^{-2/3}\right)\cup\left(-t^{-2/3}/2,\infty\right)}\frac{\step(z')-\s\left(z'\right)}{t^{2/3}z'-w}\d z',\qquad
I_2:=\int_{-2t^{-2/3}}^{-t^{-2/3}/2}\frac{-\s\left(z'\right)}{t^{2/3}z'-w}\d z'.
\ee
To estimate $I_1$, note that for $z'\in\left(-\infty,-2t^{-2/3}\right)\cup\left(-t^{-2/3}/2,\infty\right)$, we have $|t^{2/3}z'-w|\geq C$ for some $C=C(\delta)>0$, so
\begin{equation*}
|I_1|\leq\frac 1C\int_{\left(-\infty,-2t^{-2/3}\right)\cup\left(-t^{-2/3}/2,\infty\right)}|\step(z')-\s(z')|\d z'\leq
\frac 1C\int_{\R}|\step(z')-\s(z')|\d z'.
\end{equation*}
To estimate $I_2$, observe that for $|w|=1,|\Im w|<\delta$ and for $z'\in[-2t^{-2/3},-t^{-2/3}/2]$, we have $|t^{2/3}z'-w|\geq|t^{2/3}z'+1|$, such that
\begin{equation*}
|I_2|\leq
\left|\int_{-2t^{-2/3}}^{-t^{-2/3}/2}\frac{\s\left(-t^{-2/3}\right)}{t^{2/3}z'-w}\d z'\right|+\int_{-2t^{-2/3}}^{-t^{-2/3}/2}\left|\frac{\s\left(-t^{-2/3}\right)-\s\left(z'\right)}{t^{2/3}z'+1}\right|\d z',
\end{equation*}
and the first term is $t^{-2/3}\s\left(-t^{-2/3}\right)\left|\int_{-2}^{-1/2}\frac{\d w'}{w'-w}\right|=o(1)$ uniformly in the set of values of $w$ considered here (and in the sense of the boundary values at $w=-1$), while the second term is estimated by \eqref{eq:sigmaprimedecay} as
\be
\label{improved}
\int_{-2t^{-2/3}}^{-t^{-2/3}/2}\left|\frac{\s\left(-t^{-2/3}\right)-\s\left(z'\right)}{t^{2/3}z'+1}\right|\d z'
\leq
t^{-2/3}\int_{-2t^{-2/3}}^{-t^{-2/3}/2}\sup_{\xi<-t^{-2/3}/2}\left|\s'(\xi)\right|\d z'=\frac 32c_3
\ee
uniformly in $0<t<\epsilon$, with $\epsilon$ sufficiently small, such that $\epsilon^{-2/3}/2>C$, with $C$ as in \eqref{eq:sigmaprimedecay}.

With similar estimates we take care of the limit $w\to 1$ as well, and therefore we conclude that $a(w)=\O\left(t^{2/3}\right)$ uniformly on the whole unit circle $\{|w|=1\}$.
Summarizing, we have thus shown that
\begin{equation*}
\|J_S-I\|_{L^1\left(\{|w|=1\}\right)}=\O\left(t^{2/3}\right),
\quad
\|J_S-I\|_{L^2\left(\{|w|=1\}\right)}=\O\left(t^{2/3}\right),
\quad
\|J_S-I\|_{L^\infty\left(\{|w|=1\}\right)}=\O\left(t^{2/3}\right),
\end{equation*}
uniformly in $0<t<\epsilon$ (for any $\epsilon$ sufficiently small as detailed above) and in $x <  Mt^{1/3}$, for any constant $M>0$;
here the norms denote the maximum of the entrywise $L^1$-, $L^2$-, and $L^\infty$-norms.

On the other hand, for $w<-1$ or $w>1$,
\begin{equation*}
J_S\left(w\right)-I=\left(\g \chi_{\left[0,\infty\right)}\left(w\right)-\s\left(wt^{-2/3}\right)\right)\wt\Psi_{\step,-}\left(w\right)\begin{pmatrix} 0 & 1 \\ 0 & 0\end{pmatrix} \wt\Psi_{\step,-}^{-1}\left(w\right),
\end{equation*}
and {using \eqref{eq:sigmaexpdecay}} we have
\begin{equation*}
J_S\left(w\right)-I=\O\left(\e^{-c_2t^{-2/3}|w|-\frac 43 w^{3/2}+2xt^{-1/3}w^{1/2}}|w|^{1/4}\right)=
\O\left(\e^{-c_2t^{-2/3}|w|-\frac 43 w^{3/2}+2Mw^{1/2}}|w|^{1/4}\right),
\end{equation*}
showing that the contribution from the rest of the contour $\G$ is  $o\left(t^{2/3}\right)$ in the relevant regime.
\end{proof}

\begin{remark}\label{remsigmaKPZantisymmetry}
If $\sigma$ has the property $\s(-z)=1-\s(z)$ and if $\sigma'(z)$ satisfies $\left|\sigma'\left(z\right)\right|\leq c_3|z|^{-3}$ for $|z|>C$ instead of just \eqref{eq:sigmaprimedecay}, the error terms of the lemma can be strengthened to $\O(t^{4/3})$. This is clear in \eqref{eq:estimatea}, and for the limits as $w\to\pm 1$ one can estimate $I_1$ in \eqref{eq:I1I2} as
\begin{equation*}
I_1=\int_{-t^{-2/3}/2}^{t^{-2/3}/2}\frac{\s_1(z')-\s(z')}{t^{2/3}z'-w}\d z'+\int_{\left(-\infty,-2t^{-2/3}\right)\cup\left(t^{2/3}/2,+\infty\right)}\frac{\s_1(z')-\s(z')}{t^{2/3}z'-w}\d z'=\O(t^{2/3}),
\end{equation*}
by the aforementioned antisymmetry for the first term and by the exponential decay \eqref{eq:sigmaexpdecay} of $\s$ for the second term, and for $I_2$ one can proceed as in \eqref{improved}, and the stronger decay $\left|\sigma'\left(z\right)\right|\leq c_3|z|^{-3}$ for $|z|>C$ implies that $|I_2|\leq \frac 32 c_3t^{2/3}$.
\end{remark}

Therefore, the RH problem for $S$ has jumps that are small in the appropriate norms, and in a similar manner as for $Y$ in Section \ref{section:asY}, we can use this to prove that
\[
\|S_--I\|_2\leq \|(1-C_{J_S})^{-1}\|_2\|C_-\|_2\|J_S-I\|_\infty. 
\]
Moreover, \[S(w)=I+S^{(1)}(x,t)w^{-1}+\mathcal O(w^{-2}),\]
we have
\[
S^{(1)}(x,t)=\mathcal O(t^{2/3}),\qquad t\to 0_+,
\]
uniformly for $|x|\leq Mt^{1/3}$.

\subsection{Asymptotics for $u_\sigma$ and $Q_\sigma$}
The asymptotics {for $S^{(1)}(x,t)$} imply in particular, since $\wt\Psi_\s\left(w\right)=S\left(w\right)\wt\Psi_\step\left(w\right)$ for $|w|>1$ and by \eqref{eq:wtPsiasymp} that
\begin{align*}
&q_\s\left(x,t\right)=q_\step\left(x,t\right) +\O\left(1\right),\\
&p_\s\left(x,t\right)=p_\step\left(x,t\right)+\O\left(t^{1/3}\right)=\frac{x^2}{4t}-t^{-1/3}\int_{-xt^{-1/3}}^{+\infty}y_\gamma^2\left(v\right)\d v+\O\left(t^{1/3}\right),
\end{align*}
where we used \eqref{eq:idPIIp2} on the last line.
Integrating \eqref{Qp} in $x=x'$ between $-\delta$ and $x$ and using \eqref{Qp2}, we obtain
\begin{align*}
\log Q_\sigma\left(x,t\right)&=\log Q_\sigma\left(-\delta,t\right)-t^{-1/3}\int_{-\delta}^{x}\int_{-x't^{-1/3}}^{+\infty}y_\gamma^2\left(v\right)\d v\d x'+\mathcal O\left(t\right)\\
&=\log Q_\sigma\left(-\delta,t\right)-\int_{-xt^{-1/3}}^{\delta t^{-1/3}}\int_{-x't^{-1/3}}^{+\infty}y_\gamma^2\left(v\right)\d v\d x'+\mathcal O\left(t^{1/3}\right)\\
&=\log F_{\rm TW}\left(-xt^{-1/3};\gamma\right)+\mathcal O\left(t^{1/3}\right).
\end{align*}
This proves part {\it(ii)} of Theorem \ref{thmQ}.
Using \eqref{upQ}, \eqref{eq:pqrstep}, \eqref{eq:stepKdVsolution2}, and the fact that $y_\gamma\left(-xt^{-1/3}\right)$ is bounded for $|x|t^{-1/3}\leq M$, we also obtain
\begin{equation*}
u_\s\left(x,t\right)=\partial_x p_\s\left(x,t\right)=-p^2_\s\left(x,t\right)-2q_\s\left(x,t\right)=\frac x{2t}-\frac 1{t^{2/3}}y_\g^2\left(-xt^{-1/3}\right)+\O\left(1\right),
\end{equation*}
which proves part {\it (ii)} of Theorem \ref{thminitialvalue}.

\section{Small time asymptotics when $ Mt^{1/3}\leq x\leq K$}
\label{section:RH3}

In this section, we assume once more that $\sigma$ satisfies Assumptions \ref{assumptions}, and that there are constants $c_1, c_2,c_3,C>0$ such that \eqref{eq:sigmaexpdecay} and \eqref{eq:sigmaprimedecay} hold. In contrast with the discussion until this point, we now restrict to the case $\g=1$; let us note that $\s_{\rm KPZ}\left(r\right)=\left(1+\exp\left(-r\right)\right)^{-1}$ satisfies this assumption.
We denote $\s_1=\chi_{\left(0,+\infty\right)}$, in agreement with the notation $\step$ used before. For a sufficiently large positive constant $M$, we assume that $x\geq Mt^{1/3}$, and we also assume that $x$ remains bounded, say $x\leq K$ for an arbitrarily large $K$.
Our objective in this section is again to derive small $t$ asymptotics for $\Psi_\s$, this time uniform in {$Mt^{1/3}\leq x\leq K$}.
It turns out that it will still be convenient to work in the variable $w=t^{2/3}z$ introduced in Section \ref{par:changeofvariable}. However, we need to construct a local parametrix near $w=0$, different from that employed in the analysis of the previous section.

\subsection{Model RH problem needed for the local parametrix}

For the construction of a local parametrix near $w=0$, we will rely on the following model RH problem, depending on a parameter $\xi$.
 
As a side note, we mention that this RH problem is obtained by the formal substitution $\left(x,t\right)\mapsto\left(\xi,0\right)$ in the RH problem for $\Psi_\s$.

\subsubsection*{RH problem for $\Phi_\s$}
\begin{enumerate}[label=(\alph*)]
\item $\Phi_\s$ is analytic in $\mathbb C\setminus \mathbb R$.
\item $\Phi_\s$ has boundary values $\Phi_{\s,\pm}$ which are $L^2$ on any compact real set, and which are continuous on the real line except at the points $r_1, \ldots ,r_k$, and they are related by
\begin{equation}\label{jumpPhi}\Phi_{\s,+}\left(z\right)=\Phi_{\s,-}\left(z\right)\begin{pmatrix}1&1-\s\left(z\right)\\0&1\end{pmatrix}.\end{equation}
\item As $z\to\infty$ there exist functions $p_{\sigma}^0=p_{\sigma}^0\left(\xi\right), q_{\sigma}^0= q_{\sigma}^0\left(\xi\right)$ and $r_{\sigma}^0= r_{\sigma}^0\left(\xi\right)$ such that $\Phi_\s$ has the asymptotic behavior
\begin{multline}
\label{eq:Phiasymp}
\Phi_\s\left(z\right) = \left( I + \frac{1}{z} \begin{pmatrix} q_{\sigma}^0 & -\i r_{\sigma}^0 \\ \i  p_{\sigma}^0 & -q_{\sigma}^0 \end{pmatrix} + \mathcal O\left(\frac{1}{z^2}\right) \right) z^{\frac{1}{4} \sigma_3} A^{-1} {\rm e}^{\xi z^{1/2}\sigma_3}\\
\times\begin{cases}I,&|\arg z|<\pi-\delta,\\
\begin{pmatrix}1&0\\\mp 1&1\end{pmatrix},&\pi-\delta<\pm \arg z<\pi,\end{cases}
\end{multline}
for any $0<\delta<\pi/2$; here the principal branches of $z^{\frac{1}4\sigma_3}$ and $z^{1/2}$ are taken, analytic in $\mathbb C \setminus \left(-\infty,0\right]$ and positive for $z>0$. 
\end{enumerate}

The proof that this RH problem admits a solution $\Phi_\s$ for all $\xi>0$ is based on a \emph{vanishing lemma}.
It contains rather classical but technical arguments, and for that reason we defer it to Appendix \ref{appvanishinglemma}.
For now, we continue assuming that there exists a solution $\Phi_\s$, and we note that uniqueness of the solution then follows from standard arguments, completely similar to those used before for the uniqueness of $\Psi_\s$. 
Let us also mention a fact which will be useful below and which follows from standard RH theory: for any point $\xi_0>0$ there exists an open disk in the complex plane centered at $\xi_0$ such $\Phi_\s\left(z;\xi\right)$ exists for all $\xi$ in this disk and the asymptotic relation \eqref{eq:Phiasymp} is uniform for $\xi$ in this disk.

Let us note that when $\s=\s_1=\chi_{\left(0,+\infty\right)}$, the RH problem for $\Phi_\sigma$ is (a variant of) the Bessel model RH problem, and then we can construct the matrix $\Phi_{\s_1}$ explicitly in terms of Bessel functions (see \eqref{eq:Phisigma1vsPhiBessel} below): let
\be\label{eq:besselparametrix}
\Phi_{\Bessel} \left(\eta\right):=
\renewcommand*{\arraystretch}{1.25}
\begin{pmatrix}
1 & \frac{3\i}8 \\ 0 & 1\end{pmatrix}
\begin{pmatrix}
\sqrt{\pi\eta}\,{\rm I}_1\left(\sqrt{\eta}\right) & -\i \sqrt{\frac \eta\pi}\,{\rm K}_1\left(\sqrt{\eta}\right)
\\
-\i\sqrt{\pi}\,{\rm I}_0\left(\sqrt{\eta}\right) & \frac 1{\sqrt{\pi}}\,{\rm K}_0\left(\sqrt{\eta}\right)
\end{pmatrix},\quad\eta\in\C\setminus(-\infty,0],
\ee
with principal branches of the square roots, then it is well known that $\Phi_\Bessel$ satisfies the following RH problem (see e.g.\ \cite{DLMF}).

\subsubsection*{RH problem for $\Phi_\Bessel$}
\begin{enumerate}[label=(\alph*)]
\item $\Phi_\Bessel$ is analytic in $\mathbb C\setminus \left(-\infty,0\right]$.
\item $\Phi_\Bessel$ has continuous boundary values on $\left(-\infty,0\right)$, and 
\be
\Phi_{\Bessel,+}=\Phi_{\Bessel,-}\begin{pmatrix}1&1\\0&1
\end{pmatrix}.
\ee
\item As $\eta\to\infty$, $\Phi_\Bessel$ has the asymptotic behavior
\begin{multline}
\label{asPhiBessel}
\Phi_\Bessel\left(\eta\right)=\left(I+\frac{1}{\eta}
\renewcommand*{\arraystretch}{1.3}
\begin{pmatrix}-\frac 9{128} & -\frac {39\i}{512}\\ -\frac {\i}8 & \frac 9{128} \end{pmatrix}
+\O\left(\frac 1{\eta^2}\right)\right)\\
\times\ \eta^{\s_3/4}A^{-1}\e^{\eta^{1/2}\s_3}\times\begin{cases}I,&|\arg z|<\pi-\delta,\\
\begin{pmatrix}1&0\\\mp 1&1\end{pmatrix},&\pi-\delta<\pm \arg z<\pi,\end{cases}
\end{multline}
for any $0<\delta<\pi/2$; here the principal branches of $\eta^{\frac{1}4\sigma_3}$ and $\eta^{1/2}$ are taken, analytic in $\mathbb C \setminus \left(-\infty,0\right]$ and positive for $\eta>0$.
\item As $\eta\to 0$, we have $\Phi_\Bessel\left(\eta\right)=\mathcal O\left(\log\eta\right)$.
\end{enumerate}

It is then straightforward to verify that for all $\xi>0$ we have
\be
\label{eq:Phisigma1vsPhiBessel}
\Phi_{\s_1}\left(z;\xi\right)=\xi^{-\frac 12 \s_3}\Phi_\Bessel\left(\xi^2 z\right).
\ee

\subsection{Small $\xi$ asymptotics for the model RH problem}

For later use, we need to understand the behavior of $\Phi_\s(z;\xi)$ as $\xi\to 0$. To this end we use an argument parallel to that employed in Section \ref{sec:proofii}. More precisely, consider the variable $\xi^2 z=\eta$ and set 
\be
\label{eq:Xi}
\Xi_\s\left(\eta\right)=\xi^{\frac 12\s_3}\Phi_\s\left(\xi^{-2}\eta\right),
\ee
which satisfies the following RH conditions, all of which are consequences of the corresponding conditions for $\Phi_\s$.

\subsubsection*{RH problem for $\Xi_\s$}

\begin{enumerate}[label=(\alph*)]
\item $\Xi_\s$ is analytic in $\mathbb C\setminus \mathbb R$.
\item $\Xi_\s$ has boundary values $\Xi_{\s,\pm}$ on $\mathbb R$ which are $L^2$ on any compact real set, and which are continuous on the real line except at the points $\xi^2r_1, \ldots ,\xi^2r_k$, and they are related by
\begin{equation}\label{jumpXi}
\Xi_{\s,+}\left(\eta\right)=\Xi_{\s,-}\left(z\right)\begin{pmatrix}1&1-\s\left(\xi^{-2}\eta\right)\\0&1\end{pmatrix}.\end{equation}
\item As $z\to\infty$ there exist functions $p_{\sigma}^0=p_{\sigma}^0\left(\xi\right), q_{\sigma}^0=q_{\sigma}^0\left(\xi\right)$ and $r_\s^0= r_{\sigma}^0\left(\xi\right)$ such that $\Phi_\s$ has the asymptotic behavior
\begin{multline*}
\label{eq:Xiasymp}
\Xi_\s\left(\eta\right) = \left( I + \frac{1}{\eta} 
\renewcommand*{\arraystretch}{1.2}
\begin{pmatrix} \xi^2q_\s^0 & -\i\xi^3 r_\s^0 \\ \i\xi p_\s^0 & -\xi^2q_\s^0 \end{pmatrix} + \mathcal O\left(\frac{1}{\eta^2}\right) \right)
\\
\times\ \eta^{\frac{1}{4} \sigma_3} A^{-1} {\rm e}^{\eta^{1/2}\sigma_3}\times\begin{cases}I,&|\arg\eta|<\pi-\delta,\\
\begin{pmatrix}1&0\\\mp 1&1\end{pmatrix},&\pi-\delta<\pm \arg\eta<\pi,\end{cases}
\end{multline*}
for any $0<\delta<\pi/2$; here the principal branches of $\eta^{\frac{1}4\sigma_3}$ and $\eta^{1/2}$ are taken, analytic in $\mathbb C \setminus \left(-\infty,0\right]$ and positive for $\eta>0$. 
\end{enumerate}

The fact that $\s\left(\xi^{-2}\eta\right)\to\s_1\left(\eta\right)= \chi_{\left(0,+\infty\right)}\left(\eta\right)$ as $\xi\to 0_+$, uniformly away from $\xi=0$, suggests, as in Section \ref{sec:proofii}, to approximate $\Xi_\s\left(\eta\right)$ by
\begin{equation*} 
\Xi_{\s_1}\left(\eta\right)=\Phi_\Bessel\left(\eta\right)
\end{equation*}
for, say, $|\eta|>1$; $\Phi_\Bessel$ is the Bessel parametrix in \eqref{eq:besselparametrix}.
For $|\eta|<1$ we consider, in analogy with \eqref{eq:parametrix2}, the following local parametrix,
\begin{equation*}
L\left(\eta\right)=\Phi_\Bessel\left(\eta\right)\begin{pmatrix}
1 & b\left(\eta\right) \\ 0 & 1
\end{pmatrix}
\end{equation*}
with
\be
\label{eq:defb}
b\left(\eta\right)=\frac 1{2\pi\i}\left(\int_{-\infty}^0\frac{-\s\left(\eta'\xi^{-2}\right)}{\eta'-\eta}\d\eta'+\int_0^{+\infty}\frac{1-\s\left(\eta'\xi^{-2}\right)}{\eta'-\eta}\d\eta'\right).
\ee

Then we define
\be
\label{eq:defW}
W\left(\eta\right)=\begin{cases}
\Xi_\s\left(\eta\right)\Phi_\Bessel^{-1}\left(\eta\right), \; \text{if} \; |\eta|>1 \; \text{and}\; \eta\in\C\setminus\R, 
\\
\Xi_\s\left(\eta\right)L^{-1}\left(\eta\right), \; \text{if} \; |\eta|<1.
\end{cases}
\ee
Then $W$ solves the following RH problem. The proof is parallel to the proof of the RH conditions for $S$ in Section \ref{sec:proofii}, and we omit it.

\subsubsection*{RH problem for $W$}
\begin{enumerate}[label=(\alph*)]
\item $W$ is analytic in $\mathbb C\setminus \G$, where $\G$ is the oriented contour in the complex plane given by the union of $\R\cap\{|\eta|>1\}$ and $\{|\eta|=1\}$, the latter oriented counter-clockwise. See Figure \ref{figcontour}.
\item $W$ has boundary values $W_\pm$ which are continuous on $\G$; according to the chosen orientation of $\{|\eta|=1\}$, $W_+$ corresponds to the boundary value in $|\eta|<1$, and $W_-$ to that in $|\eta|>1$. Moreover, the boundary values are related by
\begin{equation*}
W_+\left(\eta\right)=W_-\left(\eta\right)J_W\left(\eta\right),
\end{equation*}
with
\begin{equation*}
J_W\left(\eta\right)=
\begin{cases}
\Phi_\Bessel\left(\eta\right)L^{-1}\left(\eta\right),&\text{if}\;|\eta|=1,
\\
\Phi_{\Bessel,-}\left(\eta\right)\begin{pmatrix} 1 & \chi_{\left[0,\infty\right)}\left(\eta\right)-\s\left(\eta\xi^{-2}\right) \\ 0 & 1\end{pmatrix} \Phi_{\Bessel,-}^{-1}\left(\eta\right),&\text{if}\;\eta<-1\;\text{or}\;\eta>1.
\end{cases}
\end{equation*}
\item As $\eta\to\infty$,  $W$ has the asymptotic behavior, uniformly in $\eta$,
\begin{equation}
W\left(\eta\right) = I + \O\left(\eta^{-1}\right).
\end{equation}
\end{enumerate}

We now show that the jump matrix $J_W$ is small in the relevant regime, provided that $\s$ satisfies, in addition to Assumptions \ref{assumptions}, the decay conditions conditions \eqref{eq:sigmaexpdecay} and \eqref{eq:sigmaprimedecay}. The following lemma is the analogue of Lemma \ref{lemma:RH2jump}.

\begin{lemma}
There exists $\epsilon>0$ such that the following estimates hold true uniformly in $0<\xi<\epsilon$;
\begin{equation*}
\|J_W-I\|_1=\O\left(\xi^2\right),
\quad
\|J_W-I\|_2=\O\left(\xi^2\right),
\quad
\|J_W-I\|_\infty=\O\left(\xi^2\right),
\end{equation*}
where $\|.\|_1$, $\|.\|_2$, and $\|.\|_\infty$ denote the maximum of the entrywise $L^1\left(\G\right)$-, $L^2\left(\G\right)$-, and $L^\infty\left(\G\right)$-norms, respectively.
\end{lemma}

\begin{proof}
On the unit circle $|\eta|=1$ we have
\begin{equation*}
J_W\left(\eta\right)-I=b\left(\eta\right)\Phi_\Bessel\left(\eta\right)
\begin{pmatrix}0 & 1 \\ 0 & 0 \end{pmatrix}\Phi_\Bessel^{-1}\left(\eta\right).
\end{equation*}
Since $\Phi_\Bessel\left(\eta\right)=\O\left(1\right)$ on the unit circle $|\eta|=1$, we have
\begin{equation*}
J_W\left(\eta\right)-I=b\left(\eta\right)\mathcal O\left(1\right).
\end{equation*}
Then we rewrite $b$ by the change of variable $z'=\eta'\xi^{-2}$ in \eqref{eq:defb}, to obtain
\begin{equation*}
|b\left(\eta\right)|\leq\frac{\xi^2}{2\pi}\int_\R\frac{|\step(z')-\s\left(z'\right)|}{|\xi^2z'-\eta|}\d z'
\leq \xi^2\frac {c_1}{|\Im \eta|}\int_\R\e^{-c_2|z'|}\d z'=\frac{2c_1}{c_2\delta}\xi^2
\end{equation*}
assuming $|\Im\eta|\geq\delta$ and using \eqref{eq:sigmaexpdecay}. This proves that $b(\eta)=\O\left(\xi^2\right)$ uniformly in $|\eta|=1,|\Im\eta|\geq\delta>0$. 
However, we can take the limit $\eta\to\pm 1$, repeating for $b(\eta)$ the argument explained after \eqref{eq:estimatea} for the analogous singular integral $a(w)$; as above, we use at this point the condition \eqref{eq:sigmaprimedecay}. We conclude therefore that $b\left(\eta\right)=\O\left(\xi^2\right)$ uniformly on the whole unit circle $\{|\eta|=1\}$.
This implies
\begin{equation*}
\|J_W-I\|_{L^1\left(\{|w|=1\}\right)}=\O\left(\xi^2\right),
\quad
\|J_W-I\|_{L^2\left(\{|w|=1\}\right)}=\O\left(\xi^2\right),
\quad
\|J_W-I\|_{L^\infty\left(\{|w|=1\}\right)}=\O\left(\xi^2\right),
\end{equation*}
where the norms are {the maxima} over the matrix entries of the $L^1$-, $L^2$-, and $L^\infty$-norms.

On the other hand, for $\eta<-1$ and for $\eta>1$,
\begin{equation*}
J_W\left(\eta\right)-I=\left(\chi_{\left[0,\infty\right)}\left(\eta\right)-\s\left(\eta \xi^{-2}\right)\right)\Phi_{\Bessel,-}\left(\eta\right)\begin{pmatrix} 0 & 1 \\ 0 & 0\end{pmatrix}\Phi_{\Bessel,-}^{-1}\left(\eta\right)
\end{equation*}
and using \eqref{asPhiBessel} and \eqref{eq:sigmaexpdecay} we have
\begin{equation*}
J_W-I=\O\left(\e^{-c_2\xi^{-2}|\eta|+\eta^{1/2}}|\eta|^{1/4}\right)
\end{equation*}
showing that the contribution from the rest of the contour $\G$ is $o(\xi^2)$ as $\xi\to 0_+$.
\end{proof}

Similarly as in the previous sections for $Y$ and $S$, writing $W(\eta)=I+W^{(1)}(\xi)\eta^{-1}+\mathcal O(\eta^{-2})$ as $\eta\to\infty$, we have as a consequence of the previous lemma that
\[W^{(1)}(\xi)=\mathcal O(\xi^2),\qquad \xi\to\infty.\]

Moreover, we can explicitly compute the sub-leading asymptotic term in the asymptotic expansion for $W^{(1)}$ as $\xi\to 0$. 
For $|\eta|=1$ we have $J_W\left(\eta\right)=I+\xi^2 J_W^{\left(1\right)}\left(\eta\right)+\O\left(\xi^4\right)$, where
\begin{equation*}
J_W^{\left(1\right)}\left(\eta\right)=\frac {I_\s}{2\pi\i\eta}\Phi_\Bessel\left(\eta\right)\begin{pmatrix}
 0 & 1 \\ 0 & 0
\end{pmatrix}\Phi_\Bessel^{-1}\left(\eta\right)
,\qquad I_\s:=\int_\R\left(\s_1\left(r\right)-\s\left(r\right)\right)\d r.
\end{equation*}
On the other parts of the jump contour we have
\[J_W(\eta)=I+\mathcal O(\xi^4).\]
Using similar arguments as the ones for $Y$ in Section \ref{section:asY}, we can show that
\be
W(\eta)=I+W_1(\eta)\xi^2+\O(\xi^4/(|\eta|+1)),\label{eq:asW}
\ee
uniformly in $\eta\in\mathbb C\setminus\Gamma$ as $\xi\to\infty$, with (substitute the above in the jump relation $W_+=W_-J_W$)
\begin{equation*}
W_{1,+}=W_{1,-}+J_W^{\left(1\right)},
\end{equation*}
which implies the identity
\begin{equation*}
W_1\left(\eta\right)=\frac 1{2\pi\i}\oint_{|\eta'|=1}\frac{J_W^{\left(1\right)}\left(\eta'\right)}{\eta'-\eta}\d\eta'.
\end{equation*}
In particular, when $|\eta|>1$, a residue computation ($J_W^{\left(1\right)}\left(\eta\right)$ is a meromorphic function of $\eta\in\C$ with a simple pole at $\eta=0$ only) gives
\be\label{eq:asW1}
W_1\left(\eta\right)=-\frac {I_\s}{2\pi\i\eta}\Phi_\Bessel\left(0\right)\begin{pmatrix}
 0 & 1 \\ 0 & 0
\end{pmatrix}\Phi_\Bessel^{-1}\left(0\right)=\frac {I_\s}\eta
\renewcommand*{\arraystretch}{1.25}
\begin{pmatrix} -\frac 3{16}&\frac{9\i}{128} \\ \frac {\i}2 & \frac 3{16}\end{pmatrix}.
\ee
Summarizing,
\be
\label{lastasymptotics}
W\left(\eta\right)=I+\frac {I_\s\xi^2}\eta
\renewcommand*{\arraystretch}{1.25}
\begin{pmatrix} -\frac 3{16}&\frac{9\i}{128} \\ \frac {\i}2 & \frac 3{16}\end{pmatrix}+\O(\xi^4/\eta)
\ee
uniformly in in $|\eta|>1$, $\xi\in\left(0,\epsilon\right]$ for all $\epsilon>0$.
Using $\Xi_\s\left(\eta\right)=W\left(\eta\right)\Phi_\Bessel\left(\eta\right)$ for $|\eta|>1$, we obtain from the $\eta\to\infty$ expansion that
\be\label{eq:asp0q0}
q_\s^0\left(\xi\right)=-\frac 9{128\xi^2}-\frac{3I_\sigma}{16}+\O\left(\xi^2\right),\qquad
p_\s^0\left(\xi\right)=-\frac 1{8\xi}+\frac{I_\sigma \xi}{2}+\O\left(\xi^3\right),\qquad \xi\to 0,
\ee
such that
\be\label{eq:asvrefined}
v_\sigma(\xi):=-p_\s^0(\xi)^2-2q_\s^0(\xi)=\frac{1}{8\xi^2}+\frac{I_\sigma}{2}+\mathcal O(\xi^2),\qquad \xi\to 0,
\ee
which proves \eqref{as:v} in Theorem \ref{thmv}.

\subsection{Asymptotics for $\widetilde{\Psi}_\sigma$}

Let us now come back to the main goal of this section, which is to establish asymptotics for the matrix $\wt\Psi_\s(w)$, introduced in \eqref{eq:defwtPsi}, in the regime $Mt^{1/3}\leq x\leq K$. To this end we shall now define suitable global and local parametrices for $\wt\Psi_\s(w)$, which are different from those considered in the previous section.

We define the global parametrix $U(w)$ by
\begin{equation*}
U(w)=w^{\sigma_3/4}A^{-1}\e^{\left(-\frac{2}{3}w^{3/2}+xt^{-1/3}w^{1/2}\right)\sigma_3}.
\end{equation*}
This function has the jump matrix $\begin{pmatrix}1&1\\0&1\end{pmatrix}$ on $(-\infty,0)$, so it satisfies approximately the same RH conditions as $\widetilde{\Psi}_\sigma$ for $w$ not too close to $0$, and it will turn out to be a good approximation for $\widetilde{\Psi}_\sigma$ for $w$ away from $0$.

For $w$ close to $0$, we need to construct a local parametrix $V$, which satisfies the same jump conditions as $\widetilde\Psi_\sigma$ and which matches with the global parametrix $U$ on a circle of, say, radius $1$ around the origin.
To this end, we define $V\left(w\right)$ for $|w|<1$ as
\begin{equation}
\label{def:P3}
V\left(w\right)=t^{\frac 16 \s_3}\Phi_\s\left(t^{-2/3}w;x-\frac 23 wt^{1/3}\right).
\end{equation}
We note here that for $\Phi_\sigma(\eta;\xi)$ to be defined, we must have that $\xi=x-\frac 23 wt^{1/3}$ lies sufficiently close to the real axis $\xi>0$; to this end we recall that we are assuming $t\to 0$ and $Mt^{1/3}\leq x\leq K$.
We now define
\begin{equation}
R\left(w\right)=\begin{cases}
(xt^{-1/3})^{\sigma_3/2}\widetilde\Psi_\sigma\left(w\right)U^{-1}\left(w\right)(xt^{-1/3})^{-\sigma_3/2}&\mbox{for $|w|>1$},\\
(xt^{-1/3})^{\sigma_3/2}\widetilde\Psi_\sigma\left(w\right)V^{-1}\left(w\right)(xt^{-1/3})^{-\sigma_3/2}&\mbox{for $|w|<1$.}
\end{cases}
\end{equation}
We will now show that $R$ solves the following RH problem.

\subsubsection*{RH problem for $R$}
\begin{enumerate}[label=(\alph*)]
\item $R$ is analytic in $\mathbb C\setminus \G$, where $\G$ is the oriented contour in the complex plane given by the union of $\R\cap\{|w|>1\}$, oriented from left to right, and the unit circle, oriented counter-clockwise. See Figure \ref{figcontour}.
\item $R$ has boundary values $R_\pm$ which are continuous on $\G$; according to the chosen orientation of $\{|w|=1\}$, $R_+$ corresponds to the boundary value in $|w|<1$, and $R_-$ to that in $|w|>1$. Moreover, the boundary values are related by
\be\label{jumpR}
R_+\left(w\right)=R_-\left(w\right)J_R\left(w\right),
\ee
with
\be
J_R\left(w\right)=
\begin{cases}
(xt^{-1/3})^{\sigma_3/2}U\left(w\right)V^{-1}\left(w\right)(xt^{-1/3})^{-\sigma_3/2},&\text{if}\;|w|=1,
\\
(xt^{-1/3})^{\sigma_3/2}U_-\left(w\right)\begin{pmatrix} 1 & \chi_{\left[0,\infty\right)}\left(w\right)-\s\left(wt^{-2/3}\right) \\ 0 & 1\end{pmatrix}\\
\qquad\qquad\qquad\times\ U_-^{-1}\left(w\right)(xt^{-1/3})^{-\sigma_3/2},&\text{if}\;w<-1\;\text{or}\;w>1.
\end{cases}
\ee
\item For $|\arg w|<\pi$, we have uniformly $$R\left(w\right) =I+ \mathcal O\left(w^{-1}\right)\; \text{as} \; w\to\infty.$$
\end{enumerate}

Condition (a) above follows from the fact that, by definition, $R$ is an analytic function of $w$ away from $\R\cup\{|w|=1\}$; it is also easily checked that $V$ and {$\widetilde\Psi_\s$} have the same jump on $\R$, and so $R$ extends analytically across $w\in\left(-1,1\right)$ (with an argument similar to the one used for $S$ to show that the points $w=t^{2/3}r_j$ and $w=0$ are removable singularities). 
To prove (b), note that the continuity of $R_\pm$ on $\G$ follows from the analogous condition for $\wt\Psi_\s$, since the points $t^{2/3}r_1, \ldots ,t^{2/3}r_k$ will be inside the disk $|w|<1$ as soon as $t$ is small enough; the jump relation follows from the computations
\begin{align*}
R_+\left(w\right)&=(xt^{-1/3})^{\sigma_3/2}\wt\Psi_\s\left(w\right) V^{-1}\left(w\right)(xt^{-1/3})^{-\sigma_3/2}\\&=(xt^{-1/3})^{\sigma_3/2}\wt\Psi_\s\left(w\right) U^{-1}\left(w\right) U\left(w\right) V^{-1}\left(w\right)(xt^{-1/3})^{-\sigma_3/2}\\&=R_-\left(w\right)(xt^{-1/3})^{\sigma_3/2}U\left(w\right) V^{-1}\left(w\right)(xt^{-1/3})^{-\sigma_3/2}
\end{align*}
for $|w|=1$, and
\begin{align*}
R_+\left(w\right)&=(xt^{-1/3})^{\sigma_3/2}\wt\Psi_{\s,+}\left(w\right)U_+^{-1}\left(w\right)(xt^{-1/3})^{-\sigma_3/2}\\&=(xt^{-1/3})^{\sigma_3/2}\wt\Psi_{\s,-}\left(w\right)\begin{pmatrix}1 & \chi_{\left[0,\infty\right)}\left(w\right)-\s\left(wt^{-2/3}\right) \\ 0 & 1 \end{pmatrix}
U_-^{-1}\left(w\right)(xt^{-1/3})^{-\sigma_3/2}
\\
&=R_-\left(w\right)(xt^{-1/3})^{\sigma_3/2}U_-\left(w\right)
\begin{pmatrix}1 & \chi_{\left[0,\infty\right)}\left(w\right)-\s\left(wt^{-2/3}\right) \\ 0 & 1 \end{pmatrix}
U_-^{-1}\left(w\right)(xt^{-1/3})^{-\sigma_3/2}
\end{align*}
for $w<-1$ and $w>1$.

Finally, condition (c) follows from the analogous condition in the RH problem for {$\wt\Psi_\s$}.

\begin{lemma}\label{lemma:JR}
As $t\to 0$ we have, uniformly for $Mt^{1/3}\leq x\leq K$ with $M>0$ large enough, for any $K>0$,
\begin{equation*}
\|J_R-I\|_1=\O\left(x^{-2}t^{2/3}\right)
,\quad
\|J_R-I\|_2=\O\left(x^{-2}t^{2/3}\right)
,\quad
\|J_R-I\|_\infty=\O\left(x^{-2}t^{2/3}\right),
\end{equation*}
where $\|.\|_1$, $\|.\|_2$, and $\|.\|_\infty$ denote the maximum of the entrywise $L^1\left(\G\right)$-, $L^2\left(\G\right)$-, and $L^\infty\left(\G\right)$-norms.
\end{lemma}

\begin{proof}
Let us first look at the jump matrix for $w<-1$ and $w>1$. 
We have
\begin{align*}
J_R(w)-I&=\left(\chi_{\left[0,\infty\right)}\left(w\right)-\s\left(wt^{-2/3}\right)\right)(xt^{-1/3})^{\sigma_3/2}U_-\left(w\right)\begin{pmatrix} 0 & 1 \\ 0 & 0\end{pmatrix} U_-^{-1}\left(w\right)(xt^{-1/3})^{-\sigma_3/2}
\\
&=\e^{-c_2t^{-2/3}|w|-\frac 43 w_-^{3/2}+2xt^{-1/3}w_-^{1/2}}\O\left(xt^{-1/3}w^{1/4}\right)
=\O\left(\e^{-\left(c_2-2Kt^{1/3}\right)t^{-2/3}|w|}xt^{-1/3}w^{1/4}\right).
\end{align*}
This implies that the contributions to the $L^1$-, $L^2$-, and $L^\infty$-norms from this part of $\G$ are of exponentially small order as $t\to 0$.

Next, we look at the jump for $R$ on the circle $|w|=1$. Here, we need to distinguish two cases: first, $\delta\leq x\leq K$ for some arbitrarily small $\delta>0$, and secondly $Mt^{1/3}\leq x\leq \delta$.
In the first case, $\xi=x-\frac{2}{3}wt^{1/3}$ tends to $x\geq \delta$ as $t\to 0_+$, such that we can use condition (c), see \eqref{eq:Phiasymp}, in the RH problem for $\Phi_\sigma$ (which holds uniformly for $\xi$ in a small neighborhood of $x$, by standard RH theory)
to obtain
\begin{align}
J_R(w)&=(xt^{-1/3})^{\sigma_3/2}U(w)V^{-1}(w) (xt^{-1/3})^{-\sigma_3/2}\nonumber\\
&=(xt^{-1/3})^{\sigma_3/2}w^{\sigma_3/4} A^{-1}\e^{\left(-\frac{2}{3}w^{3/2}+xt^{-1/3}w^{1/2}\right)\sigma_3}\Phi_\sigma^{-1}(t^{-2/3}w;\xi)t^{-\frac 16 \s_3}(xt^{-1/3})^{-\sigma_3/2}\nonumber\\
&=x^{\sigma_3/2}\left(I-\frac{t^{2/3}}{w}\begin{pmatrix}q_\s^0(\xi)&-\i r_\s^0(\xi)\\ \i p_\s^0(\xi)&-q_\s^0(\xi)
\end{pmatrix}
+\mathcal O\left(t^{4/3}\right)\right)x^{-\sigma_3/2}\label{eq:jumpRcorrection1}
\end{align}
as $t\to 0$. The corresponding result for the $L^1$-, $L^2$-,
and $L^\infty$- norms now follows easily.

In the case $Mt^{1/3}\leq x\leq \delta$, $\xi=x-\frac{2}{3}wt^{1/3}$ tends to $0$ as $t\to 0$ and we are not allowed to use \eqref{eq:Phiasymp} because these asymptotics are not valid uniformly as $\xi\to 0$.
Instead, we now use the small $\xi$ asymptotics for $\Phi_\sigma(z;\xi)$.
By \eqref{eq:Xi} and \eqref{eq:defW},
we have
\[\Phi_\sigma(t^{-2/3}w;\xi)=\xi^{-\sigma_3/2}W(\xi^2t^{-2/3}w)\Phi_{\rm Be}(\xi^2 t^{-2/3}w),\qquad |\xi^2 t^{-2/3}w|>1,\]
and moreover, from \eqref{lastasymptotics} we know that
\[W(\eta)=I+\frac {\xi^2I_\s}\eta
\renewcommand*{\arraystretch}{1.25}
\begin{pmatrix} -\frac 3{16}&\frac{9\i}{128} \\ \frac {\i}2 & \frac 3{16}\end{pmatrix}+\mathcal O(\xi^4/\eta),\] as $\xi\to 0$, uniformly for $\eta$ sufficiently large.
We now observe that for $|w|=1$, $x>Mt^{1/3}$, and $t$ sufficiently small, $|\xi^2t^{-2/3} w|>M/2$, hence we can use the asymptotics \eqref{asPhiBessel} 
and obtain
\begin{multline*}\Phi_\sigma(t^{-2/3}w;\xi)=\xi^{-\sigma_3/2}\left(I+\frac {t^{2/3}I_\s}w
\renewcommand*{\arraystretch}{1.25}
\begin{pmatrix} -\frac 3{16}&\frac{9\i}{128} \\ \frac {\i}2 & \frac 3{16}\end{pmatrix}+\mathcal O(t^{4/3})\right)\\
\times \left(I+\frac{\xi^{-2}t^{2/3}}{w}\renewcommand*{\arraystretch}{1.25}\begin{pmatrix}-\frac 9{128} & -\frac {39\i}{512}\\ -\frac {\i}8 & \frac 9{128} \end{pmatrix}+\mathcal O(\xi^{-4}t^{4/3})\right)\xi^{\sigma_3/2}t^{-\sigma_3/6}w^{\sigma_3/4}A^{-1}\e^{\xi t^{-1/3}w^{1/2}\sigma_3},\end{multline*}
as $t\to 0$, uniformly in $Mt^{1/3}\leq x\leq \delta$.
We obtain
\begin{align}
J_R(w)&=(xt^{-1/3})^{\sigma_3/2}U(w)V^{-1}(w) (xt^{-1/3})^{-\sigma_3/2}\nonumber\\
&=(xt^{-1/3})^{\sigma_3/2}w^{\sigma_3/4} A^{-1}\e^{\left(-\frac{2}{3}w^{3/2}+xt^{-1/3}w^{1/2}\right)\sigma_3}\Phi_\sigma^{-1}(t^{-2/3}w;\xi)x^{-\sigma_3/2}\nonumber\\
&=I-\frac{\xi^{-2}t^{2/3}}{w}\begin{pmatrix}-\frac 9{128} & -\frac {39\i}{512}\frac{x}{\xi}\\ -\frac {\i}8 \frac{\xi}{x}& \frac 9{128} \end{pmatrix}-\frac {t^{2/3}I_\s}w
\renewcommand*{\arraystretch}{1.25}
\begin{pmatrix} -\frac 3{16}&\frac{9\i}{128}\frac{x}{\xi} \\ \frac {\i}2 \frac{\xi}{x}& \frac 3{16}\end{pmatrix}+\mathcal O(x^{-4}t^{4/3}).\label{eq:jumpRcorrection}
\end{align}
We have proven that the contribution of the circle $|w|=1$ to the norm of $J_R-I$ is also of order $\mathcal O(x^{-2}t^{2/3})$, and this completes the proof.
\end{proof}

Similarly as before for $Y$ and for $S$, we obtain the following result for $R$ (which can be proven using the arguments from Section \ref{section:asY}):
we have $R(w)=I+R^{(1)}(x,t)w^{-1}+\mathcal O(w^{-2})$ as $w\to\infty$, with
\[R^{(1)}=\frac{1}{2\pi\i}\int_{\Gamma}(J_R(w)-I)\,\d w.\]
A closer inspection of the proof of 
Lemma \ref{lemma:JR} shows that the contribution to this integral of $\mathbb R\setminus[-1,1]$ is exponentially small in $t$ as $t\to 0$, such that the main contribution comes from the circle $|w|=1$.
By \eqref{eq:jumpRcorrection1}, we obtain 
\begin{align*}R^{(1)}(x,t)&=\frac{1}{2\pi\i}x^{\sigma_3/2}\left(t^{2/3}\begin{pmatrix}q_\s^0(\xi)&-\i r_\s^0(\xi)\\ \i p_\s^0(\xi)&-q_\s^0(\xi)
\end{pmatrix}\int_{|w|=1}\frac{\d w}{w}+\mathcal O(t^{4/3})\right)x^{-\sigma_3/2}\\
&=t^{2/3}x^{\sigma_3/2}\left(\begin{pmatrix}q_\s^0(\xi)&-\i r_\s^0(\xi)\\ \i p_\s^0(\xi)&-q_\s^0(\xi)
\end{pmatrix}+\mathcal O(t^{4/3})\right)x^{-\sigma_3/2},\end{align*}
as $t\to 0$
with $\delta\leq x\leq K$, while by \eqref{eq:jumpRcorrection}, we have 
\begin{align*}R^{(1)}(x,t)&=\frac{1}{2\pi\i}\xi^{-2}t^{2/3}\begin{pmatrix}-\frac 9{128} & -\frac {39\i}{512}\frac{x}{\xi}\\ -\frac {\i}8 \frac{\xi}{x}& \frac 9{128} \end{pmatrix}\int_{|w|=1}\frac{\d w}{w}+
\frac{t^{2/3}I_\s}{2\pi\i}
\renewcommand*{\arraystretch}{1.25}
\begin{pmatrix} -\frac 3{16}&\frac{9\i}{128}\frac{x}{\xi} \\ \frac {\i}2\frac{\xi}{x} & \frac 3{16}\end{pmatrix}
\int_{|w|=1}\frac{\d w}{w}+\mathcal O(x^{-4}t^{4/3})\\
&={\xi^{-2}t^{2/3}}\begin{pmatrix}-\frac 9{128} & -\frac {39\i}{512}\frac{x}{\xi}\\ -\frac {\i}8\frac{\xi}{x} & \frac 9{128} \end{pmatrix}+{t^{2/3}I_\s}
\renewcommand*{\arraystretch}{1.25}
\begin{pmatrix} -\frac 3{16}&\frac{9\i}{128}\frac{x}{\xi} \\ \frac {\i}2 \frac{\xi}{x}& \frac 3{16}\end{pmatrix}+\mathcal O(x^{-4}t^{4/3})\end{align*}
for $Mt^{1/3}\leq x\leq \delta$.   
Using the asymptotics for $p_\s^0, q_\s^0$ given in \eqref{eq:asp0q0}, we can write this at once as
\begin{equation*}
R^{(1)}(x,t)={t^{2/3}}x^{\sigma_3/2}\left(\renewcommand*{\arraystretch}{1.2}\begin{pmatrix}q_\s^0(\xi)&\star\\ \i p_\s^0(\xi)&-q_\s^0(\xi)
\end{pmatrix}+\mathcal O(x^{-4}t^{4/3})\right)x^{-\sigma_3/2},
\end{equation*}
where $\star$ denotes an unimportant entry.

\subsection{Asymptotics for $u_\sigma$ and $Q_\sigma$}
Plugging the above large $w$ and small $t$ asymptotics for $R$ into the identity
\[\wt\Psi_\s\left(w\right)=(xt^{-1/3})^{-\sigma_3/2}R\left(w\right)U\left(w\right)(xt^{-1/3})^{\sigma_3/2},\] valid for $|w|>1$, using \eqref{eq:wtPsiasymp}, and comparing term by term in the large $w$ expansions we obtain
\begin{equation*}
q_\s\left(x,t\right)=q_\s^0\left(\xi\right)+\O\left(x^{-4}t^{2/3}\right),\qquad p_\s\left(x,t\right)=p_\s^0\left(\xi\right)\frac{x}{\xi}+\O\left(x^{-3}t^{2/3}\right),
\end{equation*}
implying 
\be
\label{eq:uv}
u_\s\left(x,t\right)=-p^2_\s\left(x,t\right)-2q_\s\left(x,t\right)=
-p_\s^0\left(x\right)^2-2q_\s^0\left(x\right)+\O\left(x^{-3}t^{1/3}\right)=v_\s\left(x\right)\left(1+\O\left(x^{-1}t^{1/3}\right)\right),
\ee
uniformly in $Mt^{1/3}\leq x\leq K$ as $t\to 0$. 
Here we denoted
\begin{equation*}
v_\s(x)=-p_\s^0\left(x\right)^2-2q_\s^0\left(x\right),
\end{equation*}
as in \eqref{eq:asvrefined}; the properties of $v_\s(x)$ will be further studied in Section \ref{sec:v}.

To obtain asymptotics for $Q_\sigma(x,t)$, we integrate the identity \eqref{Qp} in $x$ between $rt^{1/3}$ and $x$, with $r>M$. This gives
\[\log Q_\s(x,t)=\log Q_\sigma(rt^{1/3},t)+\int_{rt^{1/3}}^x \left(p_\s(\xi,t)-\frac{\xi^2}{4t}\right)\d\xi.\]
Now we use part {\it (ii)} of Theorem \ref{thmQ} and the above asymptotics for $p_\s$ to obtain
\[\log Q_\s(x,t)=\log F_{\rm TW}(-r;1)+\int_{rt^{1/3}}^x \left(p_\s^0(\xi)-\frac{\xi^2}{4t}+\mathcal O(\xi^{-2}t^{1/3})\right)\d\xi+\mathcal O(t^{1/3}),\]
and after a straightforward computation, it follows that we have uniformly for $r$ and $xt^{-1/3}$ sufficiently large,
\begin{multline*}\log Q_\s(x,t)=\log F_{\rm TW}(-r;1)+\frac{r^3}{12}+\frac{1}{8}\log r \\-\frac{x^3}{12t}-\frac{1}{8}\log (xt^{-1/3}) +\int_{rt^{1/3}}^x \left(p_\s^0(\xi)+\frac{1}{8\xi}\right)\d\xi+\mathcal O(x^{-1}t^{1/3})+\mathcal O(r^{-1}).\end{multline*}
Now, we can rely on the known large gap asymptotics for the Tracy-Widom distribution: from \cite{DIK, BBdF}, we know that
\[\lim_{r\to +\infty}\left(\log F_{\rm TW}(-r;1)+\frac{r^3}{12}+\frac{1}{8}\log r\right)=\frac{\log 2}{24} +\log \zeta'(-1),\]
where $\zeta'$ is the derivative of Riemann's zeta function.
Hence, if we let $r\to\infty$ slow enough such that $rt^{1/3}\to 0$, we obtain
\[\log Q_\s(x,t)=-\frac{x^3}{12t}-\frac{1}{8}\log (xt^{-1/3}) +\frac{\log 2}{24} +\log \zeta'(-1)
+\int_{0}^x \left(p_\s^0(\xi)+\frac{1}{8\xi}\right)\d\xi+\mathcal O(x^{-1}t^{1/3}).\]
For this, we also used the asymptotics \eqref{eq:asp0q0} for $p_\s^0$. The integral in this expression is uniformly bounded for $x\leq K$, for any $K>0$. Using the identity $\partial_\xi p_\s^0(\xi)=v_\s(\xi)$, which we will prove below (see \eqref{eq:idpq0}),
we can also write this as
\[\log Q_\s(x,t)=-\frac{x^3}{12t}-\frac{1}{8}\log (xt^{-1/3}) +\frac{\log 2}{24} +\log \zeta'(-1)
+\int_{0}^x (x-\xi) \left(v_\s(\xi)-\frac{1}{8\xi^2}\right)\d\xi+\mathcal O(x^{-1}t^{1/3}),\]
as $t\to 0$ with $Mt^{1/3}\leq x\leq K$.

\subsection{Proof of Theorem \ref{thmv}}\label{sec:v}

We first note that \eqref{as:v} has already been proven above, see \eqref{eq:asvrefined}, as a consequence of the small $\xi$ asymptotics for the model parametrix $\Phi_\s$.

The proof of the remaining properties of $v_\s$ is parallel to the proof of \eqref{eq:u2}--\eqref{eqthm:phiAi} in Theorem \ref{theorem:main}.
By operator theoretic methods analogous to those used in Section \ref{sec:3}, one can prove that $\Phi_\s$ is differentiable in $\xi$ and that the asymptotics of $\pa_\xi\Phi_\s(z)$ for large $z$ can be obtained by formally differentiating \eqref{eq:Phiasymp}. It is then easy to show that $\wt B_0\left(z;\xi\right):=\left(\pa_x\Phi_\s\left(z;\xi\right)\right)\Phi^{-1}_\s\left(z;\xi\right)$ is a linear function of $z$;
\begin{equation*}
\wt B_0\left(z;\xi\right)=\begin{pmatrix}
p_\s^0\left(\xi\right) & \i z + 2 \i q_\s^0\left(\xi\right) \\ -\i & -p_\s^0\left(\xi\right)
\end{pmatrix},
\end{equation*}
where $p_\s^0,q_\s^0,r_\s^0$ are as in \eqref{eq:Phiasymp}.
It also follows that $\Res_{z=0}\wt B_0\left(z;\xi\right)=\Res_{z=0}\left(\pa_x\Phi_\s\left(z;\xi\right)\right)\Phi^{-1}_\s\left(z;\xi\right)=0$; the (2,1)-entry of this identity implies
\be
\label{eq:idpq0}
v_\s\left(\xi\right)=\pa_\xi p_\s^0\left(\xi\right)=-p_\s^0\left(\xi\right)^2-2q_\s^0\left(\xi\right),
\ee
in complete analogy with \eqref{upQ}.

Setting
\be
\label{Theta0}
\Theta_0\left(z;\xi\right)=\e^{\i\pi\s_3/4}\begin{pmatrix}
1 & -\i p_\s^0\left(\xi\right) \\ 0 & 1
\end{pmatrix}\Phi_\s\left(z;\xi\right)\e^{-\i\pi\s_3/4}
\ee
we have, exactly as in Proposition \ref{prop:KdVLaxpair},
\begin{equation*}
\pa_\xi\Theta_0\left(z;\xi\right)=B_0\left(z;\xi\right)\Theta_0\left(z;\xi\right),
\end{equation*}
with
\begin{equation*}
B_0\left(z;\xi\right)=\begin{pmatrix}
0 & -z+\pa_\xi p_\s^0\left(\xi\right)-p_\s^0\left(\xi\right)^2-2q_\s^0\left(\xi\right) \\ 1 & 0
\end{pmatrix}=\begin{pmatrix}
0 & -z+2v_\s\left(\xi\right) \\ 1 & 0
\end{pmatrix}.
\end{equation*}
Therefore we can write
\be
\label{Theta00}
\Theta_0\left(z;\xi\right)=\sqrt{2\pi}
\renewcommand*{\arraystretch}{1.2}
\begin{pmatrix}
\pa_\xi\psi\left(z;\xi\right) & -\pa_\xi\wt\psi\left(z;\xi\right)\\ -\psi\left(z;\xi\right) & \wt\psi\left(z;\xi\right)
\end{pmatrix}
\ee
where $\psi=\psi_\s$ and $\wt\psi=\wt\psi_\s$ solve the Schr\"odinger equation
\begin{equation*}
\pa_\xi^2\begin{pmatrix}\psi\left(z;\xi\right) \\ \wt\psi\left(z;\xi\right)\end{pmatrix}=\left(z-2v\left(\xi\right)\right)\begin{pmatrix}\psi\left(z;\xi\right) \\ \wt \psi\left(z;\xi\right)\end{pmatrix},
\end{equation*}
proving in particular \eqref{eq:psiSchrodinger}. The asymptotic relation \eqref{eq:psibessel} follows from
\begin{equation*}
-\sqrt{2\pi}\psi\left(z;\xi\right)=\left(\Theta_0\right)_{2,1}\left(z;\xi\right)\sim -\frac{\e^{\xi z^{1/2}}}{\sqrt 2 z^{1/4}}\sim-\sqrt{\pi\xi}\,{\rm I}_0\left(\xi\sqrt z\right),\quad |\arg z|<\pi-\delta,
\end{equation*}
where we used first the asymptotics for $\Phi_\s$ in \eqref{eq:Phiasymp} and \eqref{Theta0}, and then the asymptotics ${\rm I}_0\left(\zeta\right)\sim\frac{\e^{\zeta}}{\left(2\pi\zeta\right)^{1/2}}$ of the Bessel function for $\zeta\to\infty$ within $|\arg\zeta|<\frac{\pi-\delta}2$, for all $\delta>0$.

Finally, by a residue computation similar to that in the proof of Proposition \ref{prop:proofu2}, we have
\begin{equation*}
\frac 1{2\pi\i}\int_\R \Theta_0\left(z;\xi\right)\begin{pmatrix}
 0 & -\i \\ 0 & 0
\end{pmatrix}\Theta_0^{-1}\left(z;\xi\right)\d \s(r)=\Res_{z=\infty}\pa_z\Theta_0\Theta_0^{-1}.
\end{equation*}
Using \eqref{Theta00} and computing the residue at infinity by \eqref{eq:Phiasymp}, the $(2,1)$-entry of this identity gives
\begin{equation*}
\int_\R \psi^2\left(r;\xi\right)\d\s(r)=\frac \xi 2,
\end{equation*}
which proves \eqref{eq:psiintegral}. Moreover, differentiating it twice and using the Schr\"odinger equation, we get
\begin{align*}
0&=\frac 12\int_\R \pa_\xi^2\left(\psi^2\left(r;\xi\right)\right)\d\s(r)=\int_\R \left[\left(\pa_\xi\psi\left(r;\xi\right)\right)^2+\psi\left(r;\xi\right)\pa_\xi^2\psi\left(r;\xi\right)\right]\d\s(r)
\\
&=\int_\R \left[\left(\pa_\xi\psi\left(r;\xi\right)\right)^2+\left(r-2v\left(\xi\right)\right)\psi^2\left(r;\xi\right)\right]\d\s(r)=\int_\R \left[\left(\pa_\xi\psi\left(r;\xi\right)\right)^2+r\psi^2\left(r;\xi\right)\right]\d\s(r)-\xi v\left(\xi\right),
\end{align*}
which proves \eqref{eq:v1}.

\appendix

\section{Solvability of the model RH problem for $\Phi_\sigma$}\label{appvanishinglemma}

We now establish the existence of the solution to the RH problem for $\Phi_\s$; to this end we rely on the following \emph{vanishing lemma}, which shows that there are no nontrivial solutions to the \emph{homogeneous} version of the RH problem for $\Phi_\s$.

\begin{lemma}[Vanishing lemma]\label{lemmavanishing}
Let $H$ be a $2\times 2$ matrix function satisfying the following conditions.
\begin{enumerate}[label=(\alph*)]
\item $H$ is analytic in $\C\setminus\R$.
\item $H$ has boundary values $H_{\pm}$ on $\mathbb R$ which are $L^2$ on any compact real set, and which are continuous on the real line except at the points $r_1, \ldots ,r_k$, and they are related by
\be
\label{jumpH}
H_{+}\left(z\right)=H_{-}\left(z\right)\begin{pmatrix}1&1-\s\left(z\right)\\0&1\end{pmatrix}.
\ee
\item For $|\arg z|<\pi$, we have uniformly
\be
\label{normalizationH}
H\left(z\right) = \mathcal O\left(|z|^{-3/4}\right){\rm e}^{ x z^{1/2}\sigma_3} \; \text{as} \; z\to\infty.
\ee
\end{enumerate}
Then $H\left(z\right)=0$ identically in $z$.
\end{lemma}

\begin{proof}
The proof follows standard arguments, see for example \cite{DKMVZ2}. Set
\begin{equation*}
M\left(z\right)=H\left(z\right)\begin{pmatrix}
0 & -1 \\ 1 & 0
\end{pmatrix}H^\dagger\left(\overline z\right),
\end{equation*}
where $H^\dagger$ denotes the conjugate transpose of $H$.
Then $M\left( z\right)$ is an analytic function of $ z\in\C\setminus\R$, and as $ z\to\infty$ we have, from \eqref{normalizationH}, $M\left( z\right)=\O\left( |z|^{-3/2}\right)$.
Therefore by Cauchy's theorem we have
\begin{equation*}
\int_\R M_+\left( z\right)\d z=\int_\R H_+\left( z\right)\begin{pmatrix}
0 & -1 \\ 1 & 0
\end{pmatrix}H_-^\dagger\left( z\right)\d z=0
\end{equation*}
implying
\be
\label{vanishingidentity1}
\int_{-\infty}^{+\infty} H_-\left( z\right)\begin{pmatrix}
1 & 1-\s\left(z\right) \\ 0 & 1
\end{pmatrix}\begin{pmatrix}
0 & -1 \\ 1 & 0
\end{pmatrix} H_-^\dagger\left( z\right)\d z
=\int_{-\infty}^{+\infty} H_-\left( z\right)\begin{pmatrix}
1-\s\left(z\right) & -1 \\ 1 & 0
\end{pmatrix}H_-^\dagger\left( z\right)\d z=0.
\ee
Adding \eqref{vanishingidentity1} and its adjoint we obtain
\begin{multline*}
\int_{-\infty}^{+\infty} H_-\left( z\right)\begin{pmatrix}
1-\s\left(z\right) & 0 \\ 0 & 0
\end{pmatrix}H_-^\dagger\left( z\right)\d z
\\
=\int_{-\infty}^{+\infty}\left(1-\s\left(z\right)\right)\begin{pmatrix}
|H_{11,-}\left(z\right)|^2 & H_{11,-}\left(z\right)\overline{H}_{21,-}\left(z\right) \\ \overline{H}_{11,-}\left(z\right)H_{21,-}\left(z\right) & |H_{21,-}\left(z\right)|^2
\end{pmatrix}\d z=0,
\end{multline*}
implying $\left(1-\s\left(z\right)\right)H_{11,-}\left(z\right)=0$ and $\left(1-\s\left(z\right)\right)H_{21,-}\left(z\right)=0$ identically for $z$ on the real line; it follows that
\begin{equation*}
H_+\left(z\right)-H_-\left(z\right)=H_-\left(z\right)\begin{pmatrix} 0 & 1-\s\left(z\right) \\ 0 & 0\end{pmatrix}=\begin{pmatrix} 0 & \left(1-\s\left(z\right)\right)H_{11,-}\left(z\right) \\ 0 & \left(1-\s\left(z\right)\right)H_{21,-}\left(z\right)\end{pmatrix}=0,\qquad z\in\R,
\end{equation*}
and so, by the Schwartz reflection principle, $H\left(z\right)$ is analytic throughout the whole complex plane (possible poles at $r_1,\dots,r_k$ are ruled out by the $L^2$ condition on the boundary values of $H$). Therefore $H_{11}\left(z\right)$ and $H_{21}\left(z\right)$ vanish identically, because they are entire and they vanish along some nonzero interval of $\R$. It follows from this together with the jump condition \eqref{jumpH} that $H_{12}\left(z\right)$ and $H_{22}\left(z\right)$ are entire functions, which by \eqref{normalizationH} behave as $\O\left(z^{-3/4}\e^{-xz^{1/2}}\right)$ when $z\to\infty$ uniformly in all directions of the complex plane; because of our choice of branch for the square root, one concludes from the Liouville theorem that also $H_{12}\left(z\right)$ and $H_{22}\left(z\right)$ vanish identically.
\end{proof}

Finally we follow the standard argument to infer that the solution to the RH problem for $\Phi_\s$ exists; indeed, it is known that $\Phi_\s$ exists if and only if a certain Fredholm operator of index zero is invertible, and it is easy to see that the kernel of this operator is nontrivial if and only there exists a matrix $H$ with the properties listed in Lemma \ref{lemmavanishing}. We refer to \cite[App. A]{KamvissisMcLaughlinMiller2003} for a detailed explanation of this argument; in this respect we point out that the RH problem for $\Phi_\s$ can be transformed by elementary, though lengthy, transformations into an \emph{umbilical RH problem} in the sense of \cite{KamvissisMcLaughlinMiller2003}.


\paragraph{Acknowledgements.}
T.C. and G.R. were supported by the Fonds de la Recherche Scientifique-FNRS under EOS project O013018F.
M.C. was supported by the European Union Horizon 2020 research and innovation program under the Marie Sk\l odowska-Curie RISE 2017 grant agreement no. 778010 IPaDEGAN. 
The authors are grateful to Alexander Its for useful discussions.

\end{document}